\documentclass[manuscript, nonacm]{acmart}

\usepackage[labelfont=bf,labelsep=period]{caption}
\usepackage[nodayofweek]{datetime}
\usepackage{enumitem}
\usepackage{amsmath}
\usepackage{float}
\usepackage{geometry}
\usepackage{graphicx}
\usepackage{tikz}
\usetikzlibrary{patterns}
\usepackage{hyperref}
\usepackage[utf8]{inputenc}
\usepackage{numbertabbing}
\usepackage{times}
\usepackage{url}
\usepackage{listings}
\usepackage{xcolor}
\usepackage{xspace}
\usepackage{algpseudocode}
\usepackage{subcaption}
\usepackage{comment}

\usepackage{amsmath} %
\allowdisplaybreaks[2]          %
\usepackage{amssymb} %
\usepackage{amsthm} %
\usepackage{dsfont}
\newdateformat{simple}{\THEDAY\ \monthname[\THEMONTH]\ \THEYEAR}
\simple

\floatstyle{ruled}
\newfloat{algo}{htbp}{algo}
\floatname{algo}{Algorithm}

\newfloat{module}{htbp}{module}
\floatname{module}{Module}

\def \ifempty#1{\def\temp{#1} \ifx\temp\empty }

\newcommand{\str}[1]{\textsc{#1}}
\newcommand{\var}[1]{\textit{#1}}
\newcommand{\op}[1]{\textsl{#1}}
\newcommand{\msg}[2]{\ensuremath{\ifempty{#2} [\str{#1}] \else [\str{#1}, {#2}] \fi}}

\newcommand{\false}{\textsc{false}\xspace}
\newcommand{\true}{\textsc{true}\xspace}
\newcommand{\etal}{\emph{et al.}\xspace}

\newcommand{\BF}{\ensuremath{\mathbb{F}}\xspace}

\newcommand{\CF}{\ensuremath{\mathcal{F}}\xspace}
\newcommand{\CG}{\ensuremath{\mathcal{G}}\xspace}

\newcommand{\CK}{\ensuremath{\mathcal{K}}\xspace}

\newcommand{\CP}{\ensuremath{\mathcal{P}}\xspace}
\newcommand{\CQ}{\ensuremath{\mathcal{Q}}\xspace}

\newcommand\xip{\var{$x_i$}\xspace}

\newcommand\pip{\var{$p_i$}\xspace}

\newcommand\pjp{\var{$p_j$}\xspace}
\newcommand\sjp{\var{$S_j$}\xspace}
\newcommand\tjp{\var{$T_j$}\xspace}
\newcommand\xjp{\var{$x_j$}\xspace}

\begin{document}

\definecolor{mycolor}{HTML}{46C5DD}
\def\scolor{blue}

\definecolor{mycolor2}{HTML}{EF559F}
\def\scolorf{mycolor2}

\def\scolort{red}
\def\sopacity{0.5}

\title{DAG-based Consensus with Asymmetric Trust \footnotesize{(Extended Version)}}

\acmConference{PODC'25}{June 16--20, 2025}{Huatulco, Mexico}

\author{Ignacio Amores-Sesar}
\affiliation{%
  \institution{Aarhus University}
  \country{Denmark}
}

\author{Christian Cachin}
\affiliation{%
  \institution{University of Bern}
  \country{Switzerland}
}

\author{Juan Villacis}
\affiliation{%
  \institution{University of Bern}
  \country{Switzerland}
}

\author{Luca Zanolini}
\affiliation{%
  \institution{Ethereum Foundation}
  \country{United Kingdom}
}

\begin{abstract}
  In protocols with asymmetric trust, each participant is free to make its own
  individual trust assumptions about others, captured by an asymmetric quorum
  system. This contrasts with ordinary, symmetric quorum systems and with
  threshold models, where all participants share the same trust assumption.
  It is already known how to realize reliable broadcasts, shared-memory
  emulations, and binary consensus with asymmetric quorums.  In this work, we
  introduce Directed Acyclic Graph (DAG)-based consensus protocols with
  asymmetric trust. To achieve this, we extend the key building-blocks of the
  well-known DAG-Rider protocol to the asymmetric model.  Counter to
  expectation, we find that replacing threshold quorums with their asymmetric
  counterparts in the existing constant-round gather protocol does \emph{not}
  result in a sound asymmetric gather primitive. This implies that asymmetric
  DAG-based consensus protocols, specifically those based on the existence of
  common-core primitives, need new ideas in an asymmetric-trust
  model. Consequently, we introduce the first asymmetric protocol for
  computing a common core, equivalent to that in the threshold model. This leads to
  the first randomized asynchronous DAG-based consensus protocol with
  asymmetric quorums.  It decides within an expected constant number of rounds
  after an input has been submitted, where the constant depends on the quorum
  system.
\end{abstract}

\maketitle

\section{Introduction}
\label{sec:intro}

In recent years, research on consensus protocols has moved beyond serially constructed blockchains and toward \emph{Directed Acyclic Graphs (DAGs)}, which enable parallelization of transaction dissemination and ordering. Several works~\cite{DBLP:conf/podc/KeidarKNS21,DBLP:conf/ccs/SpiegelmanGSK22,DBLP:conf/eurosys/DanezisKSS22, DBLP:conf/ndss/BabelCDKKKST25} leverage DAGs to improve throughput (by concurrently batching transactions) and latency (by reducing bottlenecks inherent in a single leader or chain). Broadly, these protocols separate the dissemination of transactions from their ordering, but they differ in how clear this separation is and in the network assumptions.

For instance, Narwhal and Tusk~\cite{DBLP:conf/eurosys/DanezisKSS22} explicitly split the protocol into a mempool layer (Narwhal) responsible for fast and reliable transaction batching and dissemination, and a separate consensus layer (Tusk) that orders the batches locally based on shared randomness without extra communication. Recent protocols such as Bullshark~\cite{DBLP:conf/ccs/SpiegelmanGSK22} and Shoal~\cite{DBLP:journals/corr/abs-2306-03058, DBLP:conf/nsdi/Arun0S0S25} follow a similar approach and refine the consensus logic on top of the Narwhal DAG-based mempool to achieve partially synchronous consensus. By contrast, the pioneering asynchronous and randomized DAG-Rider~\cite{DBLP:conf/podc/KeidarKNS21} protocol integrates the consensus logic with a DAG mechanism that also disseminates transactions.

Despite the clear performance benefits of DAG-based protocols, all existing proposals rely on \emph{symmetric}, \emph{threshold} quorums among the participants. In other words, they employ global adversarial thresholds (e.g., assuming that at most \(f\) out of \(n\) participants are faulty) that are uniform across the participants. This model may fail to capture the heterogeneous and subjective trust relations that naturally arise in open and decentralized networks.

Concurrently, practical cryptocurrency networks have adopted such heterogeneous trust assumptions in their consensus protocols. In particular, the Unique Node Lists (UNLs) of Ripple~\cite{chase2018analysis,DBLP:conf/opodis/Amores-SesarCM20} and Federated Byzantine Quorum Systems of Stellar~\cite{mazieres2015stellar,DBLP:conf/sosp/LokhavaLMHBGJMM19} introduce mechanisms for the participants to locally specify which participants they deem trustworthy. In Ripple, the reliance on a recommended UNL published by one entity introduces a certain degree of centralization, however. For the Stellar protocol, ensuring safety and liveness in the presence of multiple local quorum slices poses intricate challenges~\cite{DBLP:conf/wdag/LosaGM19}.

Asymmetric quorum systems were introduced by Damgård \etal~\cite{DBLP:conf/asiacrypt/DamgardDFN07} and formalized as such by Alpos \etal~\cite{DBLP:journals/dc/AlposCTZ24}; they represent subjective trust, where the participants may have differing views of the trustworthiness of others. Their formal properties can be directly mapped to those of standard Byzantine quorum systems~\cite{DBLP:journals/dc/MalkhiR98}.
In particular, asymmetric quorum systems allow each participant to define its own conditions under which it deems other participants faulty, thereby relaxing the need for a single global adversarial assumption. This flexibility more accurately represents the varied trust relationships inherent in real-world settings, where participants may have fundamentally different views about others, in terms of reliability, security, expertise, and other attributes such as location or jurisdictions.

Existing work~\cite{DBLP:journals/dc/AlposCTZ24} presents protocols for reliable broadcasts, shared-memory emulations, and binary randomized consensus based on asymmetric quorum systems.
The protocols of the Ripple and Stellar networks demonstrate that certain forms of subjective trust can be integrated into specific consensus protocols. Their applications are limited to single-chain protocols, and a benchmark of the Stellar protocol~\cite{DBLP:conf/sosp/LokhavaLMHBGJMM19} shows much smaller throughput than that of DAG-protocols~\cite{DBLP:conf/ccs/SpiegelmanGSK22, DBLP:conf/ndss/BabelCDKKKST25}. The interplay between asymmetric trust assumptions and the concurrent processing power available in DAG-based protocols remains uncharted. 

In this paper, we introduce the first \emph{DAG-based protocols with asymmetric quorums} for computing a \emph{common core} and for randomized, asynchronous \emph{consensus}.  They expand asymmetric quorum systems to practical consensus methods that permit concurrent processing and achieve high performance.  We build on the influential DAG-Rider protocol~\cite{DBLP:conf/podc/KeidarKNS21} and change its structure to accommodate locally defined quorums. The DAG structure is employed to order transactions despite these potentially divergent views.  Inherent in DAG-Rider is a \emph{gather protocol} used to compute a \emph{common core}~\cite{canetti1996studies, DBLP:conf/wdag/AbrahamAM10}.  This primitive is weaker than yet another primitive called \textit{agreement on a core set}, however, which is equivalent to consensus.

Note that the existing broadcast and consensus protocols with asymmetric quorums~\cite{DBLP:journals/dc/AlposCTZ24} are obtained by replacing the quorum system in the standard protocols from the literature with an asymmetric quorum system.  It turns out that this heuristic fails for the standard three-round gather protocol. The combinatorial argument based on counting votes in the gather protocol does not extend to asymmetric quorums, and we exhibit a quorum system, for which the correspondingly extended standard three-round gather protocol fails to output a common core.

This leads us to formulating a novel constant-round protocol for computing a common core in the asymmetric model, which introduces extra communication steps.  Moreover, we explore other variants of the gather protocol that more closely resemble the usual ones in the threshold-quorum model.

Using the insight from the gather protocol, we then formulate a DAG-based asynchronous, randomized, and asymmetric consensus protocol.  Like the asymmetric gather protocol, it is neither following the heuristic of replacing standard quorums with asymmetric ones, but it includes additional steps.  This asymmetric consensus protocol decides within an expected constant number of rounds after an input has been submitted.  The constant itself depends on the quorum system, specifically, it is proportional to the ratio of the total number of participants to the size of the smallest quorum of any participant.

\subsection{Related Work}
\label{sec:relatedwork}

Threshold and symmetric quorums are widely used tools in distributed computing to prove the safety and liveness properties of numerous protocols~\cite{DBLP:journals/siamcomp/NaorW98,DBLP:journals/dc/MalkhiR98,AttiyaW04,DBLP:books/daglib/0025983}. They are at the heart of many practical systems in cloud computing and cryptocurrencies. One of their significant limitations is that all participants use the same quorums. This restricts the freedom of participants to choose whom to trust and leads to so-called \emph{permissioned} blockchains. The asymmetric model, introduced by Damg{\aa}rd~\etal~\cite{DBLP:conf/asiacrypt/DamgardDFN07} and further developed by Alpos~\etal~\cite{DBLP:journals/dc/AlposCTZ24}, represents a novel approach that generalizes the symmetric paradigm. It lets the participants make trust choices of their own, reflecting social connections or other information available to them outside the protocol. Namely, each participant can specify its own quorums, giving it the freedom to have its own trust assumptions. Algorithms relying on symmetric quorums cannot be mechanically extended to the asymmetric model, as no such powerful composition method exits. In particular, it could be the case that certain properties do not hold if the underlying quorums are asymmetric, and all proofs must be reassessed in the more general model. It is also necessary to redefine certain protocol properties to fit the new model, as there might be correct participants who have chosen the ``wrong friends'' and assumed too much about them.

Existing work~\cite{DBLP:journals/dc/AlposCTZ24} introduces asymmetric versions of relevant primitives like reliable broadcast, shared memory, binary consensus and a common-coin primitive. Ripple~\cite{chase2018analysis} and Stellar~\cite{mazieres2015stellar,DBLP:conf/sosp/LokhavaLMHBGJMM19} represent the two largest blockchain projects that use asymmetric trust. In Ripple, each participant must declare its own assumptions, that is, a list of other participating nodes that it trusts and from which it will consider votes \cite{chase2018analysis,DBLP:conf/opodis/Amores-SesarCM20}. Stellar uses a similar approach, where each participant keeps a list of participants it considers important and waits for a suitable majority of them to agree on a transaction before it is considered settled~\cite{mazieres2015stellar,DBLP:conf/wdag/LosaGM19,DBLP:conf/sosp/LokhavaLMHBGJMM19}. The protocols in the Stellar network differ from the established BFT-consensus algorithms like PBFT and HotStuff~\cite{DBLP:journals/tocs/CastroL02,DBLP:conf/podc/YinMRGA19}.

Directed Acyclic Graph (DAG)-based techniques represent a new approach to solving the consensus problem in a leaderless setting. Longest-chain protocols have traditionally received the most attention, but the usage of DAGs presents an opportunity to improve performance. Their decoupling of the communication and consensus layers~\cite{DBLP:conf/podc/KeidarKNS21, DBLP:conf/eurosys/DanezisKSS22}, as well as the parallelization of message proposals, increase their efficiency. Amores-Sesar~and~Cachin~\cite{DBLP:conf/esorics/Amores-SesarC24} proved that for every chain protocol $\Pi$ there is a DAG protocol $\Pi'$ that is safe and live under the same assumptions as $\Pi$, with the same or better latency and better throughput. DAG-Rider~\cite{DBLP:conf/podc/KeidarKNS21} is the pioneering protocol that proposes the usage of a certified DAG, a safe commit rule based on it, and the reliance on a common core. Unfortunately, DAG-Rider requires unbounded memory to be able to provide fairness, making it unsuitable for practical applications. Subsequent protocols, like Narwhal-Tusk\cite{DBLP:conf/eurosys/DanezisKSS22} and Bullshark~\cite{DBLP:conf/ccs/SpiegelmanGSK22} build on it, address these issues, and provide practical implementations that have also been deployed in practice.

The Sui blockchain is one of the main cryptocurrency projects that has deployed DAG-based consensus~\cite{DBLP:conf/opodis/Sonnino24,DBLP:conf/ccs/BlackshearCDKKL24}. It initially used Bullshark and subsequently switched to the Mysticeti protocol~\cite{DBLP:conf/ndss/BabelCDKKKST25}. The latter achieves consistently higher throughput than most other blockchains, reaching over 200K transactions per second in a recent update~\cite{DBLP:conf/ndss/BabelCDKKKST25}. 

Current research has focused on decreasing the latency of DAG protocols. The reliable broadcast and the wave structure used by many algorithms~\cite{DBLP:conf/podc/KeidarKNS21, DBLP:conf/ccs/SpiegelmanGSK22, DBLP:conf/eurosys/DanezisKSS22} has negative impacts on the latency. Mysticeti proposes a solution to this issue by switching from certified to uncertified DAGs, where each vertex in the graph is instead delivered through consistent broadcast. With this approach, Mysticeti also improved the latency of the Sui blockchain by a factor of four~\cite{DBLP:conf/ndss/BabelCDKKKST25}. Cordial Miners~\cite{DBLP:conf/wdag/KeidarNPS23} builds upon DAG-Rider and improves its latency by forgoing the reliable broadcast mechanism. In DAG-Rider, reliable broadcast ensures two key properties: first, that all processes receive the same set of messages, and second, that Byzantine processes cannot generate conflicting blocks. To guarantee the first property, each process maintains a set of blocks received from every other process and includes any missing blocks when sending new ones. For the second property, the protocol extends the length of the wave, enabling honest processes to detect when a Byzantine process has sent multiple blocks in a given round. With these two relatively simple yet elegant modifications, Cordial Miners creates a DAG-based protocol that reduces the latency of Bullshark by half.

\subsection{Structure of the Paper}
The remainder of this paper is organized as follows. Section~\ref{sec:preliminaries} introduces the model and reviews key concepts that serve as the foundation for our work. Section~\ref{sec:aGather} highlights the limitations of existing common core primitives in the asymmetric setting and presents the first asymmetric gather. Section~\ref{sec:adagrider}, the main contribution of this paper, introduces our asymmetric DAG-based consensus protocol and proves its correctness. Finally, conclusions are drawn in Section~\ref{sec:conclusion}. Appendix~\ref{ap:gather-counterexample} contains a (mechanical) proof for the impossibility of extending the usual gather protocol to asymmetric quorums.

\section{Preliminaries}
\label{sec:preliminaries}
\subsection{Model}

 We will consider a system of $n$ \emph{processes} $\mathcal{P}=\{p_1, \dots, p_n\}$ that interact asynchronously with each other by exchanging messages. A protocol for $\mathcal{P}$ consists of a collection of programs with instructions for all processes. They are presented using the event-based notation of Cachin~\etal~\cite{DBLP:books/daglib/0025983}. An execution starts with all processes in a special initial state; subsequently the processes repeatedly trigger events, react to events, and change their state through computation steps.  A process that follows its protocol during an execution is called \textit{correct}. A \textit{faulty} process, also called \textit{Byzantine}, may crash or deviate arbitrarily from its specification. We will assume that there is a low-level functionality for sending messages over point-to-point links between each pair of processes. This functionality is accessed through the events of \textit{sending} and \textit{receiving} a message. Point-to-point messages are authenticated and delivered reliably among correct processes. Here and from now on, the notation $A^*$ for a system $A \subseteq 2^\mathcal{P}$ , denotes the collection of all subsets of the sets in $A$, that is, $A^* = \{A' | A' \subseteq A, A \in A\}$

\subsection{Review of Symmetric Trust}
\label{sec:symmetric-trust}

Byzantine quorum systems are an essential concept in distributed systems, particularly for settings involving Byzantine faults. These systems, first introduced by Malkhi and Reiter~\cite{DBLP:journals/dc/MalkhiR98}, extend traditional quorum systems~\cite{DBLP:journals/siamcomp/NaorW98} by accounting for adversarial or arbitrary behavior among faulty processes. They are defined concerning a \emph{(symmetric) fail-prone system}, a collection of subsets of processes (\emph{fail-prone sets)} that specifies all potential failure patterns a protocol can tolerate. This framework allows protocols to achieve resilience as long as the set of actually faulty processes is a subset of these predefined fail-prone sets. A \emph{(symmetric) Byzantine quorum system} comprises a collection of \emph{quorums} (sets of processes) that satisfy two main properties: \emph{consistency}, ensuring that any two quorums intersect in at least one non-faulty process, and \emph{availability}, guaranteeing that for any fail-prone set, there exists a quorum disjoint from it. These properties enable the design of protocols that maintain correctness despite Byzantine failures.

The feasibility of Byzantine quorum systems depends on certain conditions, such as the \emph{\( Q^3 \)-condition}~\cite{DBLP:journals/dc/MalkhiR98}, which requires that no three fail-prone sets collectively cover all processes in the system. This ensures that the system can function correctly under the specified fault tolerance assumptions. 

By generalizing threshold-based fault tolerance assumptions, Byzantine quorum systems provide a robust framework for building distributed protocols that are resilient to complex failure patterns. 

\subsection{Asymmetric Trust}

In the asymmetric setting~\cite{DBLP:journals/dc/AlposCTZ24} each process is free to make its own trust assumptions. This is expressed through an \emph{asymmetric fail-prone system} $\mathbb{F} = [\mathcal{F}_1, \dots, \mathcal{F}_n]$, where $\mathcal{F}_i$ represents the trust assumptions  of process $p_i$. Each $\mathcal{F}_i$ is a collection of subsets of $\mathcal{P}$ such that some $F \in \mathcal{F}_i$ with $F \subseteq \mathcal{P}$ is called a fail-prone set for $p_i$ and contains all processes that, according to $p_i$, may at most fail together in some execution \cite{DBLP:conf/asiacrypt/DamgardDFN07}. We can in turn proceed to define the asymmetric Byzantine quorum systems as follows \cite{DBLP:journals/dc/AlposCTZ24}. 

\begin{definition}
\label{def:abqs}
An \emph{asymmetric Byzantine quorum system} for $\mathds{F}$ is an array of collections of sets $\mathbb{Q} = [\mathcal{Q}_1, \cdots, \mathcal{Q}_n]$ where $\mathcal{Q}_i \subseteq 2^{\mathcal{P}}$ for $i \in [1, n]$. The set $\mathcal{Q}_i \subseteq 2^{\mathcal{P}}$ is a symmetric quorum system of $p_i$ and any set $Q_i \in \mathcal{Q}_i$ is called a quorum for $p_i$. The system $\mathbb{Q}$ must satisfy the following two properties.

\begin{description}
    \item[Consistency:] The intersection of two quorums for any two processes contains at least one process for which either process assumes that it is not faulty, i.e.,
    $$
    \forall i, j \in [1, n], \forall Q_i \in \mathcal{Q}_i, \forall Q_j \in \mathcal{Q}_j, \forall F_{ij} \in \mathcal{F}_i^* \cap \mathcal{F}_j^*: Q_i \cap Q_j \nsubseteq F_{ij}.
    $$

    \item[Availability:] For any process $p_i$ and any set of processes that may fail together according to $p_i$, there exists a disjoint quorum for $p_i$ in $\mathcal{Q}_i$, i.e.,
    $$
        \forall i \in [1, n], \forall F_i \in \mathcal{F}_i : \exists Q_i \in \mathcal{Q}_i : F_i \cap Q_i = \emptyset.
    $$
\end{description}    
\end{definition}

Given an asymmetric quorum system $\mathbb{Q}$ for $\mathbb{F}$, an asymmetric kernel system for $\mathbb{Q}$ is defined analogously as the array $\mathbb{K} = [\mathcal{K}_1, \dots, \mathcal{K}_n]$ that consists of the kernel systems for all processes in $\mathcal{P}$. A set $K_i \in \mathcal{K}_i$ is called a \emph{kernel} for $p_i$ and for each $K_i$ it holds that $\forall Q_i \in \mathcal{Q}_i, \ K_i \cap Q_i \neq \emptyset $, that is, a kernel intersects all quorum sets of a process.

The set of processes that fail is denoted by $F$. The members of $F$ are unknown to the processes and can only be identified by an outside observer and the adversary.  A process $p_i$ correctly foresees $F$ if $F \in \mathcal{F}_i^* = \{ F'' \ | \  F'' \subseteq F', F' \in \mathcal{F} \}$, that is, $F$ is contained in one of its fail-prone sets. Based on this information it is possible to classify processes in three categories.

\begin{description}
    \item[Faulty:] a process $p_i \in F$;
    \item[Naive:] a correct process $p_i$ where $F \notin \mathcal{F}^*_i$; and
    \item[Wise:] a correct process $p_i$ where $F \in \mathcal{F}^*_i$
\end{description}

Alpos~\etal~\cite{DBLP:journals/dc/AlposCTZ24} show that naive processes might affect the safety and liveness guarantees of some protocols. In order to formalize this notion, they introduced the concept of a \textit{guild}. 
\begin{definition}
    
A guild is a set of wise processes that contains one quorum for each member. Formally, a guild $\mathcal{G}$ for $\mathds{F}$ and $\mathds{Q}$ is a set of processes that satisfies the following two properties.
\begin{description}
    \item[Wisdom:] $\mathcal{G}$ is a set of wise processes, that is,
    $$
    \forall p_i \in \mathcal{G}: F \in \mathcal{F}_i^*.
    $$

    \item[Closure:] $\mathcal{G}$ contains a quorum for each of its members, that is,
    $$
    \forall p_i \in \mathcal{G} : \exists Q_i \in \mathcal{Q}_i: Q_i \subseteq \mathcal{G}.
    $$
\end{description}
\end{definition}

The existence of asymmetric quorum systems can be characterized with a
property that generalizes the $Q^3$-condition for the underlying asymmetric
fail-prone systems as follows.
  
\begin{definition}[$B^3$-condition]
\label{def:b3}
  An asymmetric fail-prone system \BF satisfies the
  \emph{$B^3$-condition}, abbreviated as $B^3(\BF)$, whenever it holds that
  \[
    \forall i,j \in [1,n],
    \forall F_i \in \CF_i, \forall F_j\in\CF_j,
    \forall F_{ij} \in {\CF_i}^*\cap{\CF_j}^*: \,
    \CP \not\subseteq F_i \cup F_j \cup F_{ij} 
  \]
\end{definition}

The $B^3$ property of a fail-prone system $\mathbb{F}$ is intimately connected to the existence of an asymmetric quorum system for $\mathbb{F}$, as shown by the following theorem.
\begin{theorem}\label{thm:asymcanon}[Alpos et al.~\cite{DBLP:journals/dc/AlposCTZ24}]
  An asymmetric fail-prone system \BF satisfies $B^3(\BF)$ if and only if
  there exists an asymmetric quorum system for \BF.
\end{theorem}

\subsection{Gather Protocol}
\label{sec:symGather}

A \emph{gather} protocol, initially introduced by Canetti and Rabin~\cite{DBLP:conf/stoc/CanettiR93}, is used to obtain a \emph{common core set}  of values among the processes in a system. Definition~\ref{def:gather} describes it. 
\begin{definition}[Gather]
\label{def:gather}
In the gather protocol, each process \(p_i\) \textit{g-proposes}  an input, and each process must \textit{g-deliver} a set of received values and their corresponding senders, represented as pairs \((p_j, x)\), where \(x\)  is the value input by process $p_j$. The protocol satisfies the following properties:
\begin{description}
    \item[Common Core:] There exists a core set \(S^+\) of size at least \(n - f\) process-output pairs, with $f$ the failure parameter, such that all non-faulty processes include \(S^+\) in their output set.
    \item[Validity:] If a non-faulty process includes \((j, x_j)\) in its output set, where \(j\) is a non-faulty process, then \(x_j\) must be \(j\)'s input.
    \item[Agreement:] All correct processes that include a pair \((j, x)\) for any process \(j\) in their output set agree on the value \(x\). More precisely, if two non-faulty processes include the pairs \((j, x)\) and \((j, x')\) in their outputs, then \(x = x'\).
\end{description}
\end{definition}

What differentiates it from \emph{agreement on a core set (ACS)}~\cite{DBLP:conf/stoc/Ben-OrCG93}, which enforces agreement on an identical output set for all processes, is that the {gather protocol} only asks that the common core is included in the output set of each process. This distinction makes gather more efficient, operating in a constant number of rounds~\cite{DBLP:journals/dc/Ben-OrE03}, whereas ACS, being a consensus equivalent, can only provide an expected constant execution time.

Algorithm~\ref{alg:symgather} shows a simple implementation of gather taken from Abraham~\etal~\cite{DBLP:conf/wdag/AbrahamAM10}. It consists of three rounds of broadcasting and collecting values. Each process starts the execution with an initial value $x_i$, and during each round, a process collects messages coming from $n-f$ other processes. Once it has received such messages, it combines them into a new set that it will broadcast during the next round. After the third round, the process delivers the resulting set. Thanks to combinatorial properties, all processes will share at least $n-f$ common elements in their delivered sets \cite{DBLP:conf/stoc/CanettiR93}. 

\begin{algo}
  \vbox{
  \small
  \begin{numbertabbing}
    xx\=xx\=xx\=xx\=xx\=xx\=MMMMMMMMMMMMMMMMMMM\=\kill
    \textbf{state} \label{}\\
    
    \> \(S \gets \emptyset\) \label{}\\
    \> \(T \gets \emptyset\) \label{}\\
    \> \(U \gets \emptyset\) \label{}\\

    \\
    \textbf{upon}  \op{g-propose($x_i$)}  \textbf{do} \label{}\\  
        \> \op{rb-broadcast(\((\pip, \xip)\))} \label{line:symgather-init}\\
    \\

    \textbf{upon}  \op{rb-deliver((\pjp, \xjp))}  \textbf{do} \label{}\\  
        \> \( S \gets S \cup \{ (\pjp, \xjp) \}  \) \label{line:symgather-combine1}\\
    \\

    \textbf{upon} having received $n-f$ broadcasts  \textbf{do}  \label{line:symgather-round1}\\  
        \> send \( \msg{DistributeS}{\pip, S}\) to all $p \in \mathcal{P}$ \label{line:symgather-broadcasts}\\
    \\

    \textbf{upon} receiving \( \msg{DistributeS}{\pjp, \sjp}\) \textbf{do} \label{}\\  
        \> \( T \gets T \cup \sjp  \) \label{line:symgather-combine2}\\
    \\

    \textbf{upon} having received $n-f$ \(\str{DistributeS}\) messages   \textbf{do}  \label{line:symgather-round2}\\  
        \> send \( \msg{DistributeT}{\pip, T}\) to all $p_j \in \mathcal{P}$ \label{}\\
    \\

    \textbf{upon} receiving \( \msg{DistributeT}{\pjp, \tjp}\) from  \(  p_j\) \textbf{do} \label{}\\  
        \> \(U \gets U \cup \tjp\) \label{line:symgather-combine3}\\   
    \\

    \textbf{upon} having received $n-f$ \(\str{DistributeT}\) messages \textbf{do} \label{line:symgather-round3}\\  
        \> \op{g-deliver(U)} \label{line:symgather-deliver}
     \end{numbertabbing}

  }
  \caption{Symmetric Gather (process $p_i$).}
  \label{alg:symgather}
\end{algo}

Gather has many applications in distributed computing, including areas such as distributed key generation~\cite{DBLP:conf/podc/AbrahamJMMST21} and atomic broadcast protocols~\cite{DBLP:conf/podc/KeidarKNS21, DBLP:conf/ccs/SpiegelmanGSK22}. We focus on the latter, specifically on its use by the DAG-Rider protocol, which heavily relies on it to prove its liveness.

An additional property that is relevant in relation to DAG-based consensus protocols is \emph{binding common core}. A \emph{binding} gather protocol is one in which, once the first correct process has delivered the final set, the common core is fixed and cannot change. Shoup~\cite{DBLP:journals/iacr/Shoup24a} showed how an adversary can take advantage of a non-binding common core in the Tusk protocol~\cite{DBLP:conf/eurosys/DanezisKSS22} to prevent processes from delivering messages. The common core protocol of Algorithm~\ref{alg:symgather}, which is used by DAG-Rider, is not binding. To achieve this property, one extra round of communication is needed~\cite{DBLP:conf/wdag/AbrahamAM10}. DAG-Rider overcomes this issue by only revealing the common coin value once at least $2f-1$ processes have finished the gather execution.

\section{Asymmetric Gather}
\label{sec:aGather}

In this section, we introduce the concept of \emph{asymmetric gather}, a generalization of the approach discussed in Section~\ref{sec:symGather}. We present a special constant-round algorithm for the asymmetric setting in Section~\ref{sec:asymmetric-gather-constant-round}. This is motivated by two aspects. First, a constant execution time is necessary to not negatively impact the latency of any protocol building on top of the gather. Second, the counterexample presented in Section~\ref{subsec:direct-translation} shows that direct translations of constant-round common core primitives are not suitable for the asymmetric world. 

By using the quorum consistency property, it is possible to prove that the exact same structure of the threshold gather~\cite{DBLP:conf/wdag/AbrahamAM10} can be maintained in the asymmetric setting (i.e., witching rounds after receiving messages from a quorum). This comes in detriment of the number of rounds. Such an algorithm would require $\log n$ rounds to achieve the common core. This negatively impacts the latency of any protocol built on top of it, highlighting the need for a new approach  to solve the problem when using asymmetric trust.

\subsection{Protocol Definition}

In the asymmetric setting, the cardinality of sets does not inherently provide guarantees, requiring us to look beyond size to identify properties that can characterize the asymmetric gather protocol. The \emph{common core} property must be adapted, as having a set of size \( n-f \) no longer has any meaning. Furthermore, the \emph{validity} and \emph{agreement} properties must be redefined to reflect the revised characterizations of naive and wise processes.

An asymmetric gather protocol is accessed through the interfaces \textit{ag-propose} and \textit{ag-deliver}. When a process invokes \(\textit{ag-propose}(x)\), it uses \(x\) as its input to the protocol. By calling \textit{ag-deliver}, a process obtains the output set produced by the execution of the protocol. This set contains tuples \((p_j, x_j)\), indicating that party \(p_j\) has \emph{ag-proposed} $x_j$. Definition~\ref{def:agather} formally defines an asymmetric gather protocol.

\begin{definition}[Asymmetric Gather]
\label{def:agather}
A protocol for \textit{asymmetric gather} satisfies the following properties:

\begin{description}
    \item[Common Core:]
    In any execution with a guild, there exists a common set \( S^+ \) composed of the values \textit{ag-proposed} by the processes belonging to a quorum \( Q_i \in \mathcal{Q}_i \) for some process \( p_i \) in the maximal guild. For every process in the maximal guild that \textit{ag-delivers} a set \( U \), it holds that \( S^+ \subseteq U \).

    \item[Validity:]
    In any execution with a guild, if a process in the maximal guild includes \( (p_j, x_j) \) in the set it \textit{ag-delivers}, and \( p_j \) is a wise process, then \( p_j \) must have \textit{ag-proposed} \( x_j \).

    \item[Agreement:]
    In any execution with a guild, for any two sets \( U \) and \( U' \) \textit{ag-delivered} by processes in the maximal guild, if \( (p_j, x) \in U \) and \( (p_j, x') \in U' \), then \( x = x' \).
\end{description}

\end{definition}

\subsection{Common Core Primitives in the Asymmetric World}
\label{subsec:direct-translation}
Many round-based DAG consensus protocols \cite{DBLP:conf/podc/KeidarKNS21, DBLP:conf/eurosys/DanezisKSS22, DBLP:conf/ccs/SpiegelmanGSK22} rely on common core primitives like gather for their commit rules. Their usage guarantees that after several rounds of receiving and broadcasting vertices there will exist a common set of vertices in the local DAGs of all processes. DAG-Rider and Bullshark employ the gather protocol~\cite{DBLP:conf/stoc/CanettiR93}, which requires 3 rounds, while Tusk uses a simpler 2 round common core primitive. 

In order to obtain asymmetric equivalents of such consensus algorithms, the development of an asymmetric common core primitive is a necessary step. The standard approach to obtain asymmetric translations of symmetric (threshold) algorithms is to replace threshold quorums (i.e., sets of size more than $\frac{n+f}{2}$) with asymmetric quorums and threshold kernels with asymmetric kernels. This technique is employed by Alpos \etal~\cite{DBLP:journals/dc/AlposCTZ24} to obtain asymmetric versions of consistent broadcast, reliable broadcast, common coin and binary consensus.

For example, they demonstrate how Bracha broadcast~\cite{DBLP:journals/iandc/Bracha87} can be seamlessly extended to an asymmetric setting. Bracha's protocol ensures reliable broadcast under Byzantine faults by utilizing symmetric trust. In this protocol, a process delivers a message \( m \) after receiving \( 2f + 1 \) \str{Ready} messages containing \( m \), in line with the symmetric model’s threshold Byzantine quorum system. Furthermore, if a process receives \( f + 1 \) \str{Ready} messages for \( m \) but has not yet sent a \str{Ready} message, it does so to ensure totality. For asymmetric quorums, these conditions are generalized to require the message to be received from a quorum and a kernel for \( p_i \) respectively.

A natural starting step to obtain an asymmetric version of gather is to follow the same approach. Instead of each process waiting to receive messages from $n-f$ processes, as in the original algorithm \cite{DBLP:conf/wdag/AbrahamAM10}, each process would wait to receive messages from a set of processes $\mathcal{P} \setminus  F_i$ for one of its fail-prone sets $F_i \in \mathcal{F}_i$. If this approach worked, a process $p_i$ would have a common core set $S^+$ with all other processes after four rounds. This attempt to obtain an asymmetric gather is shown in Algorithm~\ref{alg:asymgather-try}. Unfortunately, as shown by Lemma~\ref{lem:agather-no-common-core}, the nature of asymmetric quorums hinders this outcome. The combinatorial arguments used to prove the protocol in the threshold setting are not applicable in the asymmetric world. The lemma is proven in Appendix~\ref{ap:gather-counterexample}. By using the quorum consistency property it can be shown that such approach requires a logarithmic number of rounds to reach a common core, which is undesirable in our setting. Until this point all asymmetric protocols have been obtained by replacing quorums, this is the first protocol where the standard approach fails.

\begin{algo}
  \vbox{
  \small
  \begin{numbertabbing}
    xx\=xx\=xx\=xx\=xx\=xx\=MMMMMMMMMMMMMMMMMMM\=\kill
    \textbf{state} \label{}\\
    
    \> \(S \gets \emptyset\) \label{}\\
    \> \(T \gets \emptyset\) \label{}\\
    \> \(U \gets \emptyset\) \label{}\\

    \\
    \textbf{upon}  \op{ag-propose($x_i$)}  \textbf{do} \label{}\\  
         \op{arb-broadcast(($\pip, \xip$))} \` // Broadcast interface of an asymmetric reliable             broadcast protocol \label{}\\
    \\

     \textbf{upon}  \op{arb-deliver((\pjp, \xjp))}  \textbf{do} \` // Deliver interface of an asymmetric reliable broadcast protocol \label{}\\  
        \> \( S \gets S \cup \{ (\pjp, \xjp) \}  \) \label{line:asymgathertry-combine1}\\
    \\

    \textbf{upon} \( \exists Q_i \in \CQ_i, \forall \pjp \in Q_i:  \) \op{arb-delivered} from \(\pjp\) \textbf{do}  \label{line:asymgathertry-round1}\\  
        \> send \( \msg{DistributeS}{\pip, S}\) to all $p \in \mathcal{P}$ \label{}\\
    \\

    \textbf{upon} receiving \( \msg{DistributeS}{\pjp, \sjp}\) \textbf{do} \label{}\\  
        \> \( T \gets T \cup \sjp  \) \label{line:asymgathertry-combine2}\\
    \\

    \textbf{upon} \( \exists Q_i \in \CQ_i, \forall \pjp \in Q_i:  \) received $\str{DistributeS}$ message from \(\pjp\)  \textbf{do} \label{line:asymgathertry-round2}\\  
        \> send \( \msg{DistributeT}{\pip, T}\) to all $p_j \in \mathcal{P}$ \label{}\\
    \\

    \textbf{upon} receiving \( \msg{DistributeT}{\pjp, \tjp}\) from  \(  p_j\) \textbf{do} \label{}\\  
        \> \(U \gets U \cup \tjp\) \label{line:asymgathertry-combine3}\\   
    \\

     \textbf{upon} \( \exists Q_i \in \CQ_i, \forall \pjp \in Q_i:  \) received $\str{DistributeT}$ message from \(\pjp\) \textbf{do} \label{line:asymgathertry-round3}\\  
        \> \op{ag-deliver($U$)} \label{line:asymgathertry-deliver}
     \end{numbertabbing}

  }
  \caption{Asymmetric Gather attempt using quorum replacements (process $p_i$).}
  \label{alg:asymgather-try}
\end{algo}

\begin{lemma}
    \label{lem:agather-no-common-core}
    Algorithm~\ref{alg:asymgather-try} does not satisfy the common core property
\end{lemma}

Consider the fail-prone system defined in Figure~\ref{fig:ce-failprone}. It represents a system with 30 processes, each of them having only one fail-prone set. Each process is represented by a row, where the processes in its fail-prone set are those colored in striped red in the row. The associated canonical quorums (i.e., any process not in the fail-prone set is contained in the quorum) is represented by processes colored blue. This fail-prone system satisfies the $B^3$ condition (Definition~\ref{def:b3}). It is not a realistic example, but it shows how a careful selection of asymmetric quorums can lead to an execution of Algorithm~\ref{alg:asymgather-try} where no common core is achieved. The complete execution is displayed in Appendix~\ref{ap:gather-counterexample}. We also provide a small code snippet to verify the results. 

The counterexample requires 30 processes. This high amount is necessary as a consequence of the quorum consistency property. After executing Algorithm~\ref{alg:asymgather-try} any system having less than 16 processes will always satisfy the common core property. Since any two quorums of any two processes will intersect in at least one process, after executing the 3 rounds of Algorithm~\ref{alg:asymgather-try}, there will be a common core among all processes. If the algorithm were modified to have $n$ rounds, then by the same reasoning any system with less than $2^n$ processes would have a common core after the execution. This justifies the need for such a large and complex counterexample. 

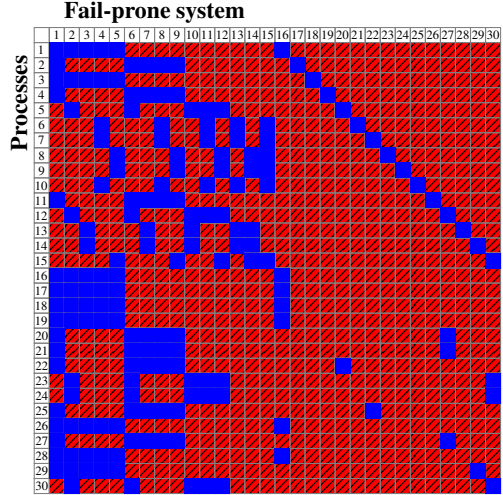
\begin{figure}[ht]
\centering
\begin{tikzpicture}[scale=0.20]
    \def\rows{30}
    \def\cols{30}

    \foreach \x in {1,...,\cols} {
        \foreach \y in {0,...,29} {
            \fill[preaction={fill, red, fill opacity=\sopacity}, fill opacity=\sopacity, pattern=north east lines] (\x, \y) rectangle (\x+1, \y+1);
        }
    }

    \foreach \x in {0,...,\cols} {
        \foreach \y in {0,...,\rows} {
            \draw[gray] (\x, \y) rectangle (\x+1, \y+1);
        }
    }

    \foreach \x in {1,...,30} {
        \pgfmathsetmacro{\result}{int(31- \x)} 
        \node[align=center, outer sep=0pt] at (\x + 0.5, 30.5) { \fontsize{6}{0}\selectfont \textcolor{black}{\scalebox{.8}{\x}}};
    }
    \foreach \y in {1,...,30} {
        \pgfmathsetmacro{\result}{int(31- \y)} 
        \node[align=center, text width=1cm] at (0.5, \y - 0.5) {\fontsize{6}{0}\selectfont \textcolor{black}{\scalebox{.8}{\result}}};
    }

\def\yy{1}
\pgfmathsetmacro{\y}{int(30- \yy)}
\foreach \x in {16, 1, 2, 3, 4, 5} {
    \fill[white] (\x, \y) rectangle (\x+1, \y+1);
    \draw[gray] (\x, \y) rectangle (\x+1, \y+1);
    \fill[\scolor, fill opacity=\sopacity] (\x, \y) rectangle (\x+1, \y+1);
}
\def\yy{2}
\pgfmathsetmacro{\y}{int(30- \yy)}
\foreach \x in {1, 17, 6, 7, 8, 9} {
    \fill[white] (\x, \y) rectangle (\x+1, \y+1);
    \draw[gray] (\x, \y) rectangle (\x+1, \y+1);
    \fill[\scolor, fill opacity=\sopacity] (\x, \y) rectangle (\x+1, \y+1);
}
\def\yy{3}
\pgfmathsetmacro{\y}{int(30- \yy)}
\foreach \x in {1, 2, 3, 4, 5, 18} {
    \fill[white] (\x, \y) rectangle (\x+1, \y+1);
    \draw[gray] (\x, \y) rectangle (\x+1, \y+1);
    \fill[\scolor, fill opacity=\sopacity] (\x, \y) rectangle (\x+1, \y+1);
}
\def\yy{4}
\pgfmathsetmacro{\y}{int(30- \yy)}
\foreach \x in {1, 19, 6, 7, 8, 9} {
    \fill[white] (\x, \y) rectangle (\x+1, \y+1);
    \draw[gray] (\x, \y) rectangle (\x+1, \y+1);
    \fill[\scolor, fill opacity=\sopacity] (\x, \y) rectangle (\x+1, \y+1);
}
\def\yy{5}
\pgfmathsetmacro{\y}{int(30- \yy)}
\foreach \x in {2, 20, 6, 10, 11, 12} {
    \fill[white] (\x, \y) rectangle (\x+1, \y+1);
    \draw[gray] (\x, \y) rectangle (\x+1, \y+1);
    \fill[\scolor, fill opacity=\sopacity] (\x, \y) rectangle (\x+1, \y+1);
}
\def\yy{6}
\pgfmathsetmacro{\y}{int(30- \yy)}
\foreach \x in {4, 21, 8, 11, 13, 15} {
    \fill[white] (\x, \y) rectangle (\x+1, \y+1);
    \draw[gray] (\x, \y) rectangle (\x+1, \y+1);
    \fill[\scolor, fill opacity=\sopacity] (\x, \y) rectangle (\x+1, \y+1);
}
\def\yy{7}
\pgfmathsetmacro{\y}{int(30- \yy)}
\foreach \x in {4, 22, 8, 11, 13, 15} {
    \fill[white] (\x, \y) rectangle (\x+1, \y+1);
    \draw[gray] (\x, \y) rectangle (\x+1, \y+1);
    \fill[\scolor, fill opacity=\sopacity] (\x, \y) rectangle (\x+1, \y+1);
}
\def\yy{8}
\pgfmathsetmacro{\y}{int(30- \yy)}
\foreach \x in {5, 23, 9, 12, 14, 15} {
    \fill[white] (\x, \y) rectangle (\x+1, \y+1);
    \draw[gray] (\x, \y) rectangle (\x+1, \y+1);
    \fill[\scolor, fill opacity=\sopacity] (\x, \y) rectangle (\x+1, \y+1);
}
\def\yy{9}
\pgfmathsetmacro{\y}{int(30- \yy)}
\foreach \x in {5, 24, 9, 12, 14, 15} {
    \fill[white] (\x, \y) rectangle (\x+1, \y+1);
    \draw[gray] (\x, \y) rectangle (\x+1, \y+1);
    \fill[\scolor, fill opacity=\sopacity] (\x, \y) rectangle (\x+1, \y+1);
}
\def\yy{10}
\pgfmathsetmacro{\y}{int(30- \yy)}
\foreach \x in {4, 8, 25, 11, 13, 15} {
    \fill[white] (\x, \y) rectangle (\x+1, \y+1);
    \draw[gray] (\x, \y) rectangle (\x+1, \y+1);
    \fill[\scolor, fill opacity=\sopacity] (\x, \y) rectangle (\x+1, \y+1);
}
\def\yy{11}
\pgfmathsetmacro{\y}{int(30- \yy)}
\foreach \x in {1, 6, 7, 8, 9, 26} {
    \fill[white] (\x, \y) rectangle (\x+1, \y+1);
    \draw[gray] (\x, \y) rectangle (\x+1, \y+1);
    \fill[\scolor, fill opacity=\sopacity] (\x, \y) rectangle (\x+1, \y+1);
}
\def\yy{12}
\pgfmathsetmacro{\y}{int(30- \yy)}
\foreach \x in {2, 6, 27, 10, 11, 12} {
    \fill[white] (\x, \y) rectangle (\x+1, \y+1);
    \draw[gray] (\x, \y) rectangle (\x+1, \y+1);
    \fill[\scolor, fill opacity=\sopacity] (\x, \y) rectangle (\x+1, \y+1);
}
\def\yy{13}
\pgfmathsetmacro{\y}{int(30- \yy)}
\foreach \x in {3, 7, 10, 28, 13, 14} {
    \fill[white] (\x, \y) rectangle (\x+1, \y+1);
    \draw[gray] (\x, \y) rectangle (\x+1, \y+1);
    \fill[\scolor, fill opacity=\sopacity] (\x, \y) rectangle (\x+1, \y+1);
}
\def\yy{14}
\pgfmathsetmacro{\y}{int(30- \yy)}
\foreach \x in {29, 3, 7, 10, 13, 14} {
    \fill[white] (\x, \y) rectangle (\x+1, \y+1);
    \draw[gray] (\x, \y) rectangle (\x+1, \y+1);
    \fill[\scolor, fill opacity=\sopacity] (\x, \y) rectangle (\x+1, \y+1);
}
\def\yy{15}
\pgfmathsetmacro{\y}{int(30- \yy)}
\foreach \x in {5, 30, 9, 12, 14, 15} {
    \fill[white] (\x, \y) rectangle (\x+1, \y+1);
    \draw[gray] (\x, \y) rectangle (\x+1, \y+1);
    \fill[\scolor, fill opacity=\sopacity] (\x, \y) rectangle (\x+1, \y+1);
}
\def\yy{16}
\pgfmathsetmacro{\y}{int(30- \yy)}
\foreach \x in {16, 1, 2, 3, 4, 5} {
    \fill[white] (\x, \y) rectangle (\x+1, \y+1);
    \draw[gray] (\x, \y) rectangle (\x+1, \y+1);
    \fill[\scolor, fill opacity=\sopacity] (\x, \y) rectangle (\x+1, \y+1);
}
\def\yy{17}
\pgfmathsetmacro{\y}{int(30- \yy)}
\foreach \x in {16, 1, 2, 3, 4, 5} {
    \fill[white] (\x, \y) rectangle (\x+1, \y+1);
    \draw[gray] (\x, \y) rectangle (\x+1, \y+1);
    \fill[\scolor, fill opacity=\sopacity] (\x, \y) rectangle (\x+1, \y+1);
}
\def\yy{18}
\pgfmathsetmacro{\y}{int(30- \yy)}
\foreach \x in {16, 1, 2, 3, 4, 5} {
    \fill[white] (\x, \y) rectangle (\x+1, \y+1);
    \draw[gray] (\x, \y) rectangle (\x+1, \y+1);
    \fill[\scolor, fill opacity=\sopacity] (\x, \y) rectangle (\x+1, \y+1);
}
\def\yy{19}
\pgfmathsetmacro{\y}{int(30- \yy)}
\foreach \x in {16, 1, 2, 3, 4, 5} {
    \fill[white] (\x, \y) rectangle (\x+1, \y+1);
    \draw[gray] (\x, \y) rectangle (\x+1, \y+1);
    \fill[\scolor, fill opacity=\sopacity] (\x, \y) rectangle (\x+1, \y+1);
}
\def\yy{20}
\pgfmathsetmacro{\y}{int(30- \yy)}
\foreach \x in {1, 6, 7, 8, 9, 27} {
    \fill[white] (\x, \y) rectangle (\x+1, \y+1);
    \draw[gray] (\x, \y) rectangle (\x+1, \y+1);
    \fill[\scolor, fill opacity=\sopacity] (\x, \y) rectangle (\x+1, \y+1);
}
\def\yy{21}
\pgfmathsetmacro{\y}{int(30- \yy)}
\foreach \x in {1, 6, 7, 8, 9, 27} {
    \fill[white] (\x, \y) rectangle (\x+1, \y+1);
    \draw[gray] (\x, \y) rectangle (\x+1, \y+1);
    \fill[\scolor, fill opacity=\sopacity] (\x, \y) rectangle (\x+1, \y+1);
}
\def\yy{22}
\pgfmathsetmacro{\y}{int(30- \yy)}
\foreach \x in {1, 20, 6, 7, 8, 9} {
    \fill[white] (\x, \y) rectangle (\x+1, \y+1);
    \draw[gray] (\x, \y) rectangle (\x+1, \y+1);
    \fill[\scolor, fill opacity=\sopacity] (\x, \y) rectangle (\x+1, \y+1);
}
\def\yy{23}
\pgfmathsetmacro{\y}{int(30- \yy)}
\foreach \x in {2, 6, 10, 11, 12, 30} {
    \fill[white] (\x, \y) rectangle (\x+1, \y+1);
    \draw[gray] (\x, \y) rectangle (\x+1, \y+1);
    \fill[\scolor, fill opacity=\sopacity] (\x, \y) rectangle (\x+1, \y+1);
}
\def\yy{24}
\pgfmathsetmacro{\y}{int(30- \yy)}
\foreach \x in {2, 6, 10, 11, 12, 30} {
    \fill[white] (\x, \y) rectangle (\x+1, \y+1);
    \draw[gray] (\x, \y) rectangle (\x+1, \y+1);
    \fill[\scolor, fill opacity=\sopacity] (\x, \y) rectangle (\x+1, \y+1);
}
\def\yy{25}
\pgfmathsetmacro{\y}{int(30- \yy)}
\foreach \x in {1, 6, 7, 8, 9, 22} {
    \fill[white] (\x, \y) rectangle (\x+1, \y+1);
    \draw[gray] (\x, \y) rectangle (\x+1, \y+1);
    \fill[\scolor, fill opacity=\sopacity] (\x, \y) rectangle (\x+1, \y+1);
}
\def\yy{26}
\pgfmathsetmacro{\y}{int(30- \yy)}
\foreach \x in {16, 1, 2, 3, 4, 5} {
    \fill[white] (\x, \y) rectangle (\x+1, \y+1);
    \draw[gray] (\x, \y) rectangle (\x+1, \y+1);
    \fill[\scolor, fill opacity=\sopacity] (\x, \y) rectangle (\x+1, \y+1);
}
\def\yy{27}
\pgfmathsetmacro{\y}{int(30- \yy)}
\foreach \x in {1, 6, 7, 8, 9, 27} {
    \fill[white] (\x, \y) rectangle (\x+1, \y+1);
    \draw[gray] (\x, \y) rectangle (\x+1, \y+1);
    \fill[\scolor, fill opacity=\sopacity] (\x, \y) rectangle (\x+1, \y+1);
}
\def\yy{28}
\pgfmathsetmacro{\y}{int(30- \yy)}
\foreach \x in {16, 1, 2, 3, 4, 5} {
    \fill[white] (\x, \y) rectangle (\x+1, \y+1);
    \draw[gray] (\x, \y) rectangle (\x+1, \y+1);
    \fill[\scolor, fill opacity=\sopacity] (\x, \y) rectangle (\x+1, \y+1);
}
\def\yy{29}
\pgfmathsetmacro{\y}{int(30- \yy)}
\foreach \x in {1, 2, 3, 4, 5, 29} {
    \fill[white] (\x, \y) rectangle (\x+1, \y+1);
    \draw[gray] (\x, \y) rectangle (\x+1, \y+1);
    \fill[\scolor, fill opacity=\sopacity] (\x, \y) rectangle (\x+1, \y+1);
}
\def\yy{30}
\pgfmathsetmacro{\y}{int(30- \yy)}
\foreach \x in {2, 6, 10, 11, 12, 30} {
    \fill[white] (\x, \y) rectangle (\x+1, \y+1);
    \draw[gray] (\x, \y) rectangle (\x+1, \y+1);
    \fill[\scolor, fill opacity=\sopacity] (\x, \y) rectangle (\x+1, \y+1);
}

\node at (8, 32) {\textbf{Fail-prone system}}; %
\node[rotate=90] at (-1, 26) {\textbf{Processes}}; %

\end{tikzpicture}
    \caption{Fail-prone system that leads to no common core in an asymmetric execution of gather. Each process (row) has only one fail-prone set, which is represented by the processes in striped red on its row. Processes in blue represent the canonical quorum associated to said fail-prone set. }
    \label{fig:ce-failprone}
\end{figure}

The same counterexample can be used to show how an asymmetric translation of Tusk~\cite{DBLP:conf/eurosys/DanezisKSS22} following the quorum replacement idea reaches no common core. The failure of existing common core primitives to work in the asymmetric world, as portrayed by these two cases, serves as motivation for the development of new equivalent asymmetric protocols. We now show a new algorithm that overcomes the problems portrayed here and obtains an asymmetric gather.

\subsection{Algorithm}
\label{sec:asymmetric-gather-constant-round}
Algorithm~\ref{alg:core} presents a novel solution to the asymmetric gather problem in a constant number of rounds. It circumvents the issues shown in Section~\ref{subsec:direct-translation} by slightly modifying the round-based structure of the symmetric gather protocol. This is done by introducing a series of control messages that prevent processes from proceeding with the algorithm until other processes satisfy certain conditions. It follows a structure similar to Algorithm~\ref{alg:symgather}, where each party builds sets $S, T$, and $U$ in progressive rounds and $U$ is output as the common core. The proof of correctness of Algorithm~\ref{alg:core} is presented in Section~\ref{sec:agather-proof}.

Each process $p_i$ starts by creating a candidate common core set $S_i$ containing initial values broadcast by one of its quorums~(line~\ref{line:creates}). After the execution of the algorithm, the $S_i$ set of one of the processes will be the common core. 

Once the $S_i$ set is created, each process will try to distribute its $S_i$ set to all other participants. If a process distributes its $S_i$ set among all members of one of its quorums $Q$, it becomes much easier to distribute it to the whole system later. Suppose $p_i$ successfully distributed its $S_i$ set to one of its quorums $Q$. Now, if any other process queries one of its quorums $Q'$ for the candidate $S$ sets they have received, at least one process in $Q'$ is in possession of $S_i$. This follows from the quorum consistency property which guarantees that $Q'\cap Q \neq \emptyset$.

The $S_i$ sets are sent to all other processes using a $\str{DistributeS}$ message~(line~\ref{line:sends}).
When a party $p_i$ receives an $S$ set, it first waits until it has \textit{arb-delivered} the component messages included in $S$. This is of importance for the validity and agreement properties. Then it adds its contents to a set used to collect all received $S$ sets, denoted $T_i$ set (line~\ref{line:createt}). At a later point the $T$ sets will be distributed to all other processes, in a manner similar to Algorithm~\ref{alg:symgather}. When receiving an $S$ set the process also sends an acknowledgment back to the sender confirming that it has received the set and added it to $T_i$ (line~\ref{line:sendack}). A process receives an acknowledgment  only if its $S$ set is actually included in the $T$ set of the receiver. If a party receives a $S$ set but its $T$ set has already been distributed, then the received set is ignored and no acknowledgment is sent back. We want to be sure that at least one process in the maximal guild has successfully distributed its $S$ set among one of its quorums before processes send their $T$ sets. If this is the case, then we can be sure that all processes will receive this candidate $S$ set when they receive $T$ sets coming from one of their quorums in the next round. 

For this purpose, the processes send a $\str{Ready}$ message once they have received acknowledgments back from one of their quorums. The purpose of a $\str{Ready}$ message is to inform other processes that the sender has successfully distributed its $S$ set among one of its quorums. If a wise process receives $\str{Ready}$ messages from one of its quorums then it can be sure that at least one process in the maximal guild has distributed its $S$ set. Ideally, once a process has received $\str{Ready}$ messages from a quorum, it would send its $T$ set through a $\str{DistributeT}$ message. This would in turn cause this process to stop sending acknowledgments in response to received $\str{DistributeS}$ messages. Unfortunately, if the process is a member of many quorums in the maximal guild and if not enough processes have distributed their $S$ sets in one of their quorums yet, stopping sending acknowledgments could affect the liveness of other processes. In order to address this issue it is necessary to delay the sending of $\str{DistribtueT}$ messages until it can be guaranteed that all processes in the maximal guild will have received enough messages to send their $\str{DistributeT}$ messages. 

We use the quorum amplification technique from Bracha Broadcast~\cite{DBLP:journals/iandc/Bracha87} to ensure that enough processes send their $T$ sets. Instead of a process distributing its $T$ set after receiving $\str{Ready}$ messages from a kernel, it sends a $\str{Confirm}$ message. It is not until a process has received such messages from a quorum (line~\ref{line:send_ready2}) that it distributes its $T$ set and stops acknowledging received $\str{DistributeS}$ messages. Processes can also send $\str{DistributeT}$ messages after receiving $\str{Confirm}$ messages from a kernel. We prove in Lemma~\ref{lem:distribute_t} that this is enough to prove that all processes in the maximal guild will send their $T$ sets, which contain candidate $S$ sets received from some processes. We know that this will guarantee that at least one candidate $S$ set will be possessed by all processes belonging to the maximal guild. This $S$ set is the common core we are looking for.

\begin{algo}
  \vbox{
  \small
  \begin{numbertabbing}
    xx\=xx\=xx\=xx\=xx\=xx\=MMMMMMMMMMMMMMMMMMM\=\kill
    \textbf{state} \label{}\\
    
    \> \(S_i \gets \emptyset\) \label{}\\
    \> \(T_i \gets \emptyset\) \label{}\\
    \> \(U_i \gets \emptyset\) \label{}\\
    \> \(\var{sentT} \gets \false\) \label{}\\

    \\
    \textbf{upon}  \op{ag-propose($x_i$)}  \textbf{do} \label{line:init}\\  
        \> \op{arb-broadcast((\(\pip, \xip\)))} \label{}\\
    \\

    \textbf{upon}  \op{arb-deliver((\pjp, \xjp))}  \textbf{do} \label{line:deliver}\\  
        \> \( S_i \gets S_i \cup \{ \text{($\pjp, \xjp$)} \}  \) \label{line:creates}\\
    \\

    \textbf{upon} \( \exists Q_i \in \CQ_i, \forall \pjp \in Q_i:  \pjp \in S_i  \)   \textbf{do}  \label{line:init_quorums}\\  
        \> send \( \msg{DistributeS}{\pip, S_i}\) to all $p_j \in \mathcal{P}$ \label{line:sends}\\
    \\

    \textbf{upon} receiving \( \msg{DistributeS}{\pjp, \sjp}\) such that $\sjp \subseteq S_i$ and $\neg \var{sentT}$ \textbf{do} \label{line:receives}\\  
        \> \( T_i \gets T_i \cup \sjp  \) \label{line:createt}\\
        \> send \( \msg{\str{Ack}}{\pip}\) to \(\pjp\) \label{line:sendack}\\  
    \\

    \textbf{upon} \( \exists Q_i \in \CQ_i, \forall \pjp \in Q_i:  \) received \( \msg{\str{Ack}}{\pjp}\) from \(\pjp\) \textbf{do} \label{line:sendready}\\  
        \> send \( \msg{Ready}{}\) to all $p_j \in \mathcal{P}$ \label{}\\ 
    \\

    \textbf{upon} \( \exists Q_i \in \CQ_i, \forall \pjp \in Q_i:  \) received \( \msg{\str{Ready}}{}\) from \(\pjp\) \textbf{do} \label{line:send_ready2}\\  
        \> send \( \msg{Confirm}{}\) to all $p_j \in \mathcal{P}$ \label{line:sendready2}\\
    \\

    \textbf{upon} \( \exists K_i \in \CK_i, \forall \pjp \in K_i:  \) received \( \msg{Confirm}{}\) from \(\pjp\) \textbf{do} \label{line:send_ready2_kernel}\\  
        \> send \( \msg{Confirm}{}\) to all $p_j \in \mathcal{P}$ \label{line:sendready22}\\
    \\

    \textbf{upon} \( \exists Q_i \in \CQ_i, \forall \pjp \in Q_i:  \) received \( \msg{Confirm}{}\) from \(\pjp\) \textbf{do} \label{line:sendt}\\  
        \> send \( \msg{DistributeT}{\pip, T}\) to all $p_j \in \mathcal{P}$ \label{}\\
        \> $\var{sentT} \gets \true$ \label{line:sentt_true}\\
    \\

    \textbf{upon} receiving \( \msg{DistributeT}{\pjp, \tjp}\) from  \(  p_j\) such that $\tjp \subseteq S_i$ \textbf{do}  \label{line:core-receivet}\\  
        \> \(U \gets U \cup \tjp\) \label{}\\   
    \\

    \textbf{upon}  receiving \( \msg{DistributeT}{\pip, T}\) from a quorum \(Q_i \in \CQ_i\) \textbf{do}  \label{line:finish}\\  
        \> \op{ag-deliver(U)} \label{}
     \end{numbertabbing}

  }
  \caption{Constant-round asymmetric gather (process $p_i$).}
  \label{alg:core}
\end{algo}

\subsection{Proof of Correctness of Algorithm~\ref{alg:core}}
\label{sec:agather-proof}
We begin by proving that in all executions with a guild at least one process in the maximal guild will send its $T$ set.

\begin{lemma}
    In any execution with a guild, a process in the maximal guild will send a $\str{DistributeT}$ message.
    \label{lem:wise_sendt}
\end{lemma}
\begin{proof}
We proceed by contradiction. Suppose no process in the maximal guild ever sends a $\str{DistributeT}$ message. 
If this happens then it follows that no process will ever receive a quorum of $\str{Confirm}$ messages (line \ref{line:sendt}). Consequently, the \var{sentT} variable will always remain false (line \ref{line:sentt_true}). As a result, whenever a process in the maximal guild receives a $\str{DistributeS}$ message, it will always respond with an $\str{Ack}$. Given that the maximal guild includes quorums for all its members, every process in the guild will receive $\str{Ack}$ messages from a quorum, prompting them to subsequently send a $\str{Ready}$ message (line \ref{line:sendready}).

Now, since all processes in the maximal guild have sent $\str{Ready}$ messages, each process will also receive $\str{Ready}$ messages from a quorum and proceed to send $\str{Confirm}$ messages (line \ref{line:send_ready2}). Consequently, every process in the maximal guild will receive $\str{Confirm}$ messages from a quorum, causing them to send a $\str{DistributeT}$ message. This scenario leads to a contradiction. Therefore, at least one process in the maximal guild must send a $\str{DistributeT}$ message.
\end{proof}

We now prove an auxiliary lemma with respect to the sending of $\str{Confirm}$ messages. It will be useful later to prove that all processes in the maximal guild will eventually send a $T$ set.

\begin{lemma}
In all executions with a guild, if a process $p_i$ in the maximal guild sends a $\str{Confirm}$ message, then there exists a process in the maximal guild that has received $\str{Ready}$ messages from at least one of its quorums.
    \label{lem:ready_from_quorum}
\end{lemma}
\begin{proof}
Without loss of generality, consider the first process $p_i$ in the maximal guild that sends a $\str{Confirm}$ message. According to Algorithm~\ref{alg:core}, this occurs either after receiving a quorum of $\str{Ready}$ messages (line~\ref{line:send_ready2}) or after receiving $\str{Confirm}$ messages from a kernel (line \ref{line:send_ready2_kernel}). If the first case holds, we have already proven what we aim to show.

Otherwise, it means that a kernel of the process has sent $\str{Confirm}$ messages. Since a process in the maximal guild has at least one quorum consisting solely of members of the maximal guild, we can deduce that at least one process in the kernel that sent a $\str{Confirm}$ message is also a member of the maximal guild. However, since no process in the maximal guild has yet sent a $\str{Confirm}$ message at this point in time, we can rule out this case for $p_i$. Therefore, $p_i$ must have received $\str{Ready}$ messages from one of its quorums.
\end{proof}

The following lemma proves that there exists a candidate $S$ set for a process in the maximal guild that will be received by all members of a quorum for that process.

\begin{lemma}
In any execution with a guild, the processes in one of the quorums $Q_i \in \CQ_i$ for a process $p_i$ in the maximal guild will receive a $\str{DistributeS}$ message from $p_i$ before they send a $\str{DistributeT}$ message.
    \label{lem:distributes}
\end{lemma}
\begin{proof}
From Lemma~\ref{lem:wise_sendt}, we know that at least one process in the maximal guild will send a $\str{DistributeT}$ message. This occurs after receiving $\str{Confirm}$ messages from a quorum. Since all quorums contain at least one wise process (Lemma 7~\cite{DBLP:journals/dc/AlposCTZ24}), we can deduce that at least one wise process has sent a $\str{Confirm}$ message. Applying Lemma~\ref{lem:ready_from_quorum}, we conclude that at least one process $p_j \in \CG_{max}$ has received $\str{Ready}$ messages from the processes in one of its quorums $Q_j \in \CQ_j$.

Thus, at least one process $p_i$ in the maximal guild has sent a $\str{Ready}$ message, which means that it received $\str{Ack}$ messages in response to its $\str{DistributeS}$ messages from one of its quorums. Since $\str{Ack}$ messages are only sent if the $\str{DistributeT}$ messages have not yet been sent (line \ref{line:receives}), we can infer that at least all processes in one of the quorums $Q_i$ of $p_i$ will receive a $\str{DistributeS}$ message from $p$ before they send a $\str{DistributeT}$ message.

The $T$ sets sent together with the $\str{DistributeT}$ messages include all the $S$ sets received up to that point. Therefore, we know that when the processes in $Q_i$ send their $\str{DistributeT}$ messages, they will contain the $S$ set sent by $p_i$.
\end{proof}

Finally, we prove that all processes in the maximal guild will send their $T$ set. This helps prove the liveness of the algorithm.
\begin{lemma}
    In any execution with a guild, all processes in the maximal guild will send a $\str{DistributeT}$ message.
    \label{lem:distribute_t}
\end{lemma}
\begin{proof}
    From Lemma~\ref{lem:wise_sendt} we know that at least one process $p_i$ in the maximal guild will send a $\str{DistributeT}$ message. Therefore this process received $\str{Confirm}$ messages from one of its quorums $Q_i \in \CQ_i$. Consider any other process $p_j \in \CG_{max}$. Since  $p_i$ and $p_j$ are both wise, it holds $F \in \CF_i^*$ and $F \in \CF_j^*$, which implies $F \in \CF_i^* \cap \CF_j^*$. Then, the set $K=Q_i \setminus F$ intersects every quorum of $p_j$ by quorum consistency and therefore contains a kernel for $p_j$. Since $K$ consists only of correct processes all of them have sent $\str{Confirm}$ messages also to $p_j$ and $p_j$ eventually sends $\str{Confirm}$ messages as well (line~\ref{line:send_ready2_kernel}). 

    We now know that all processes in the maximal guild will send a $\str{Confirm}$ message. Therefore all processes in the maximal guild are guaranteed to receive $\str{Confirm}$ messages from at least one of their quorums (since the maximal guild contains a quorum for every member). Following Algorithm~\ref{alg:core} this means that all processes in the maximal guild will send a $\str{DistributeT}$ message (line~\ref{line:sendt}). 
\end{proof}

We prove in the next lemma that the common core property will be achieved for all processes in the maximal guild. We do this using the previously proved Lemmas \ref{lem:distributes} and \ref{lem:distribute_t}. 
\begin{lemma}
    In any execution with a guild, after the execution of Algorithm~\ref{alg:core} all processes in the maximal guild will have a set $S_i$ created by a process $p_i$ in the maximal guild within their $U$ sets. 
    \label{lem:commoncore}
\end{lemma}
\begin{proof}
    From Lemma~\ref{lem:distributes} we know that a process $p_i \in \CG_{max}$ will distribute its $S$ set $S_i$ among one of its quorums $Q_i \in \CQ_i$.

    Consider any process $p_j \in \CG_{max}$.  From Lemma~\ref{lem:distribute_t} we know that all processes in the maximal guild will send a $\str{DistributeT}$ message and therefore it is guaranteed that $p_j$ will receive $\str{DistributeT}$ messages from at least one quorum $Q_j \in \CQ_j$. 

    By quorum consistency we know that $Q_i \cap Q_j \neq \emptyset$ and therefore one of the processes $p$ from which $p_j$ receives a $\str{DistribtuteT}$ message will also belong to $Q_i$. Since all processes in $Q_i$ contain the set $S_i$ within their $T$ sets $p_j$ will also receive $S_i$ when it receives the $\str{DistributeT}$ message from $p$. Therefore, for all processes $p_k$ in the maximal guild, it will hold $S_i \subseteq U_k$.

    This shows that $S_i$ will be the common core. 
\end{proof}

Finally, we show that Algorithm~\ref{alg:core} satisfies all properties needed for an asymmetric gather protocol (Definition~\ref{def:agather}).
\begin{lemma}
    Algorithm~\ref{alg:core} implements an asymmetric gather protocol. 
\end{lemma}
\begin{proof}
    From Lemma~\ref{lem:commoncore} we know that Algorithm~\ref{alg:core} satisfies the common core property.

    Validity and agreement follow from the use of the asymmetric reliable broadcast primitive to distribute the initial messages (line \ref{line:deliver}). Since all correct processes wait to deliver the original input messages before accepting any $S$ or $T$ set that contains them (Lines~\ref{line:receives}~and~\ref{line:core-receivet}), it is guaranteed that all processes in the maximal guild receive the same messages with the same values.
\end{proof}

\section{Asymmetric DAG-based Consensus}
\label{sec:adagrider}

The ideas of the asymmetric gather primitive presented in Section~\ref{sec:aGather} can be used to obtain an asymmetric version of the DAG-Rider algorithm \cite{DBLP:conf/podc/KeidarKNS21}. 

\subsection{DAG-based Consensus}

Directed Acyclic Graph (DAG)-based consensus protocols represent a more efficient approach to achieving consensus than traditional chain-based protocols~\cite{DBLP:conf/esorics/Amores-SesarC24}. In these protocols, the total ordering of delivered messages emerges naturally from the structure of the directed acyclic graphs on which they are based, with the aid of a common coin but without the need for extra communication. 

DAG-Rider, proposed by Keidar~\etal~\cite{DBLP:conf/podc/KeidarKNS21}, is a pioneering algorithm to incorporate these techniques. Its elegant wave structure has influenced other protocols such as Bullshark~\cite{DBLP:conf/ccs/SpiegelmanGSK22}, and Tusk~\cite{DBLP:conf/eurosys/DanezisKSS22}. For these reasons, DAG-Rider is the obvious starting point for creating an asymmetric DAG-based protocol. 

Its authors propose an atomic broadcast implementation based on three primitives: a common coin, reliable broadcast, and a gather protocol. They divide time into rounds. During a round $r$, each process creates a vertex containing a block of new messages and invokes $\op{aa-broadcast}$ by reliably broadcasting it. Each new vertex also contains edges pointing to the vertices that the process received from other processes during round $r-1$. The fact that a vertex for round $r$ references the vertices received during the previous round leads to the emergence of the DAG structure. 

Once the vertex has been created, the process reliably broadcasts it to the other processes and waits to deliver $2f+1$ vertices for round $r$ broadcast by other processes. Upon delivering such vertices, the process considers the round as finished and moves forward to the next round, where it performs equivalent actions. After four rounds of exchanging vertices, referred to as a \emph{wave} in the protocol, one process is chosen as the leader using a common coin. The vertex created by the leader process during the first round of the wave is considered the starting point for message delivery. The algorithm then takes this vertex and atomically delivers all the messages contained within it in a deterministic manner. The algorithm proceeds similarly for all vertices reachable through directed paths starting at the leader, following a deterministic order. Once a wave is reached, where all messages have already been atomically delivered (i.e., ``committed"), the algorithm stops.

This algorithm is proven~\cite{DBLP:conf/podc/KeidarKNS21} to satisfy the properties of Byzantine atomic broadcast. Since during a round, each process only waits to deliver $2f+1$ vertices before proceeding to the next, it is possible that the local DAG built by different processes does not have the same vertices. To ensure all participating processes deliver the same set of messages, the leader vertex must be contained in the local DAG built by each process. The algorithm uses the gather protocol to lower-bound the probability that this condition holds. By construction of DAG-Rider, each wave corresponds to a gather execution. The vertex of round 1 corresponds to the proposed value and each vertex in rounds 2, 3, and 4 corresponds to the sets created by the processes during the respective rounds in the gather protocol.

After the execution of a gather, we know that all correct processes will have a (unknown to them) common core within their output sets. In the DAG-Rider algorithm, we achieve a similar guarantee. Specifically, there will be a set of round 1 vertices (the ``common core") such that all processes that completed the wave have them in their local DAG. This ensures that if the common coin selects a vertex in the common coin, all processes have the vertex in their local DAG and can commit the messages for this wave. Since the common core has size $n-f$, the selected vertex is in the common core with probability at least $\frac{2}{3}$. If the chosen vertex does not belong to the common core, the algorithm delays committing the messages until a successful wave.

\subsection{Asymmetric DAG-based Consensus}
In the following, we present an adaptation of DAG-Rider~\cite{DBLP:conf/podc/KeidarKNS21} for the asymmetric world. To derive this algorithm, we rely on the asymmetric gather protocol introduced in Section~\ref{sec:aGather}, as well as the common coin and reliable broadcast primitives presented by Alpos~\etal~\cite{DBLP:journals/dc/AlposCTZ24}. 

Atomic Broadcast is a powerful consensus primitive \cite{DBLP:journals/iandc/CristianASD95}, and many approaches to solving the consensus problem involve developing an algorithm to address the atomic broadcast problem \cite{DBLP:conf/podc/KeidarKNS21, DBLP:conf/osdi/CamaioniGMRVV24}. We will follow the same approach in the asymmetric setting.

\begin{definition}[Asymmetric Atomic Broadcast]
\label{def:core}
A protocol for \textit{asymmetric} atomic broadcast is defined through two events \textit{aa-broadcast} and \textit{aa-deliver} and satisfies the following properties
\begin{description}
    \item[Agreement:] In all executions with a non-empty guild, if one process in the maximal guild delivers a message then all other processes in the maximal guild will eventually deliver it with probability 1
    \item[Validity:] In all executions with a non-empty guild, if a correct process broadcasts a message then all processes in the maximal guild must eventually deliver it with probability 1
    \item[Total order:] In all executions with a non-empty guild, if any process in the maximal guild receives message $m$ before message $m'$ then every other process in the maximal guild must receive message $m$ before $m'$
    \item[Integrity:] In all executions with a non-empty guild, any process in the maximal guild delivers a message at most once
\end{description}
\end{definition}

\subsection{Algorithm}
\label{sec:asymdagrider_alg}

Algorithms~\ref{alg:adag}, \ref{alg:adag2}~and~\ref{alg:adag3} presents the first DAG-based consensus algorithm using asymmetric trust. It is built by adapting the main building blocks of DAG-Rider to the asymmetric setting. We use the asymmetric reliable broadcast and common coin introduced by Alpos~\etal~\cite{DBLP:journals/dc/AlposCTZ24} and employ the asymmetric gather introduced in Section~\ref{sec:aGather}. The gather is not used in a black-box manner, as a sub-protocol called by the algorithm, but instead each wave of Algorithm~\ref{alg:adag} is built following the same structure as Algorithm~\ref{alg:core}. That is, each wave is an asymmetric gather execution. This guarantees that in the last round of the wave, all processes have a common core of first-round vertices. This is used to prove liveness, in the same way as in DAG-Rider. The proof of correctness of Algorithm~\ref{alg:adag} is presented in Appendix~\ref{ap:adag-proof}.

As is the case in the original DAG-Rider, the algorithm consists of continuous repetitions of gather-like waves. Each wave is composed of four rounds, each of which follows the same structure as Algorithm~\ref{alg:core} and where participants create a DAG data structure. The symmetric and asymmetric versions differ primarily in two key aspects: the rule governing round changes and the commit rule.

In {DAG-Rider}, a round concludes once a process receives {\(2f+1\) vertices} broadcast by other processes. However, adapting this to the {asymmetric setting}, where no intrinsic properties arise from set cardinalities, requires a new rule:

\begin{quote}
{A round is considered complete when a process \(p_i\) has delivered vertices from all members of at least one of its quorums \(Q_i \in \CQ_i\).}
\end{quote}

This idea is the same as in Algorithm~\ref{alg:core}, where processes complete $S$ and $U$ sets after receiving previous round messages from one of their quorums. It will apply to round 1 and round 3 sets. In addition to this condition, round 2 sets will be treated slightly differently, as they are the $T$ set equivalent. In Algorithm~\ref{alg:core} they are distributed once a quorum of $\str{Confirm}$ messages has been received. Algorithm~\ref{alg:adag} will follow the same approach and processes, in addition to waiting for the round-change rule, will also wait to receive $\str{Confirm}$ messages from a quorum. 

In what follows, we will denote by $Q_i \models arr$ the fact that array $arr$ contains messages sent by all processes belonging to a quorum $Q_i$, e.g., $Q_i \models DAG[r]$ if $DAG[r]$ contains vertices broadcast by all processes in $Q_i$. We will also use the term \emph{message} or \emph{transaction} to refer  to the blocks sent in the algorithm.

A strong path between two vertices in the DAG is a path where all edges are between vertices in consecutive rounds. In the symmetric DAG-Rider the commit rule establishes that there must be strong paths from at least $2f+1$ vertices belonging to the fourth round of a wave towards the leader for the wave to be committed. As is proven in their paper, this guarantees that this vertex will be reachable from any future leader through a strong path. In the case of {asymmetric DAG-Rider}, this rule is modified to adapt it to the asymmetric setting. Specifically:

\begin{quote}
A process \(p_i\) commits a leader if there exists a quorum \(Q_i \in \CQ_i\) for $p_i$ such that all vertices in \(DAG_i[\text{round}(w, 4)]\) proposed by members of \(Q_i\) have a strong path to the leader.
\end{quote}

We prove in Lemma~\ref{lem:adag-leaderpath} that this new rule guarantees that any future leader will have a strong path to the leader of the current wave.

As the algorithm is a succession of identical waves, it is enough to explain how one wave develops to grasp the main ideas behind it. In what follows we explain a wave $w$ composed of rounds 1, 2, 3 and 4. Each wave follows the structure of an asymmetric gather. During round 1 each process creates a new vertex containing messages that should be atomically broadcast (line~\ref{line:newv1}). It then reliably broadcasts the vertex and waits until $Q_i \models DAG[1]$. This is equivalent to creating a candidate $S$ set in Algorithm~\ref{alg:core}. Afterwards, it moves on to the second round.

In round 2 the process again creates a new vertex, which has edges to all vertices in the previous round on its local DAG, and reliably broadcast it. This is equivalent to sending the $\str{DistributeS}$ messages in Algorithm~\ref{alg:core} with edges representing the inclusion of values in the set. Upon broadcasting the vertex the process will also deliver vertices broadcast by other processes. Following the outline of Algorithm~\ref{alg:core}, the process will be ready to move to the next round after it has received $\str{Confirm}$ messages from a quorum. Every time the process delivers a vertex for round 2, it replies with an acknowledgment to the sender (line~\ref{line:ack_vertex}). Once a process has received $\str{Ack}$ messages from a quorum, it sends a $\str{Ready}$ message to all other processes (line~\ref{line:ack_quorum}). Upon receiving $\str{Ready}$ messages from a quorum the process sends $\str{Confirm}$ messages to all processes (line~\ref{line:ready_quorum}). Finally, upon receiving $\str{Confirm}$ messages from a quorum (line~\ref{line:ready2_quorum}) the process is ready to continue moving forward to round 3. This is signalized by setting the value of variable $tReady$ to $\true$. Setting this variable allows the process to create a new vertex, which references all vertices received during the previous round, and reliably broadcast it to all other processes (line~\ref{line:round2}). This follows the flow of control messages sent in Algorithm~\ref{alg:core} and has the same intention of preventing moving to the next round until the second round messages have been distributed to sufficient processes. In the next round, the process reliably broadcasts the vertex it created and waits to receive round 3 vertices from one of its quorums. Now that we have finished executing all the rounds, and because we have followed the same structure as that of Algorithm~\ref{alg:core}, we know that there exists a common core among the vertices in the first round of the $DAG$. Knowing this, we can attempt the delivery of the messages

This process starts by selecting a leader vertex from the first round of the wave using a common coin (line~\ref{line:getleader}), which guarantees that all processes in the maximal guild will select the same leader. Given that a process changes rounds after having delivered vertices coming from a quorum, it is possible that it does not have all vertices sent during a round in its local DAG. This implies that a process might not yet have delivered the leader of a round. If this is the case, the process \textit{skips} this wave and does not commit any vertex. If the process has the leader vertex in its local DAG it checks if the commit rule is satisfied. If it is, then the wave is committed. This implies starting a graph traversal at the leader, visiting all reachable vertices, and delivering the blocks of messages they contain. This routine continues until reaching vertices in a wave that has already been committed. Lemma~\ref{lem:adag-leaderpath} shows that if some process commits a leader, then all other processes, even those that do not currently have the leader in their local DAG (i.e., they will deliver this vertex later), will be able to commit this vertex and all other uncommitted vertices reachable from it. Therefore if some process in the maximal guild commits a wave then eventually all other processes will commit it and atomically deliver its associated messages in the same order.

In the symmetric DAG-Rider, a process commits, in expectation, every $\frac{3}{2}$ waves. Since each wave has a constant size, due to the constant-round duration of the gather, the number of time units between commits is expected constant. Lemma~\ref{lem:rider-expected} shows that in Algorithm~\ref{alg:adag}, a commit occurs every $\frac{|\mathcal{P}|}{c(\mathbb{Q})}$ waves, where $c(\mathbb{Q})$ is the size of the smallest quorum in the system. Since quorums are fixed throughout an execution, this value is constant. Each wave in Algorithm~\ref{alg:adag} also has constant duration, that of Algorithm~\ref{alg:core}. As a result, the number of time units between commits in Algorithm~\ref{alg:adag} is also expected constant. This emphasizes the importance of using a constant-round asymmetric gather.

\begin{algo}
  \vbox{
  \small
  \begin{numbertabbing}
    xx\=xx\=xx\=xx\=xx\=xx\=MMMMMMMMMMMMMMMMMMM\=\kill
    \textbf{state} \label{}\\
    \> \(r \gets 0\) \` // current round counter \label{}\\
    \>  \(\var{DAG} \gets \{\}\) \` // array of set of vertices \label{}\\
    \> \> \(\var{DAG}[0] \gets \text{hardcoded quorum of vertices for $p_i$}\) \label{}\\
    \> \> \(\var{DAG}[j] \gets \{\}, \text{for all $j>0$} \) \label{}\\
    \> \var{blocksToPropose} \( \gets \{\} \)  // stores valid trans. from clients to broadcast \label{}\\
    \> \var{deliveredVertices} \( \gets \{\}\) \` // stores vertices already delivered \label{}\\
    \> \var{leadersStack} \( \gets \{\}\) \`  \label{}\\
    \> \var{tReady} \( \gets \false\) \label{}\\
    \> \var{buffer}\( \gets \{\}\) \` // store vertices until they are processed \label{}\\
    \> \var{ack}\( \gets [\bot]^N\) \` // \str{Ack} msgs. from other processes \label{}\\
    \> \var{ready}\( \gets [\bot]^N\) \` // \str{Ready} msgs. from other processes \label{}\\
    \> \var{confirm} \( \gets [\bot]^N\) \` // \str{Confirm} msgs. from processes \label{}\\
    \> \var{decidedWave} \( \gets 0\) \` // last committed wave \label{}\\

    \\

    \textbf{procedure}  \op{createNewVertex}(\var{round}) \label{}\\  
    \> \textbf{wait until} \(\neg \var{blocksToPropose.empty()}\) \label{}\\
    \> \(\var{v.block} \gets \var{blocksToPropose.dequeue()}\) \label{}\\
    \> \(\var{v.strongEdges} \gets \var{DAG[round-1]}\) \label{adag:set-strong-edges}\\
    \> \(\op{setWeakEdges(\var{v}, \var{round})}\) \label{}\\
    \> \textbf{return} \var{v} \label{}\\
    \\

    \textbf{procedure}  \op{setWeakEdges(\var{v}, \var{round})}  \label{}\\  
    \> \(\var{v.weakEdges} \gets \{\}\) \label{}\\
    \> \textbf{for} $\var{r}=\var{round}-2, \dots, 1$ \textbf{do} \label{}\\
    \> \> \textbf{for} $\var{u} \in \var{DAG[r]}$ s.t. $\neg \op{path(\var{v}, \var{u})}$ \textbf{do} \label{}\\
    \> \> \> \(\var{v.weakEdges} \gets \var{v.weakEdges} \cup \{\var{u}\}\) \label{}\\
    \\

    \textbf{procedure}  \op{getWaveVertexLeader(\var{w})}  \label{}\\  
    \> \(p_j \gets \op{chooseLeader}_i(w)\) \label{}\\
    \> \textbf{if} \(\exists v \in \op{DAG}[\op{round}(w, 1)]\) : \(v.source=p_j\) \textbf{then} \label{}\\
    \> \> \textbf{return} $v$ \label{}\\
    \> \textbf{return} \(\bot\)  \label{}\\
    \\

    \textbf{while}  \true \textbf{do}  \label{line:adag-while}\\  
        \> \textbf{for} \(\var{v} \in \var{buffer}: \var{v.round} \leq \var{r}\) \textbf{do} \label{}\\
        \> \> \textbf{if} \(\forall \var{v}' \in \var{v.strongEdges} \cup \var{v.weakEdges}: \var{v}' \in \bigcup_{k\geq 1} \op{\var{DAG[k]}} \) \textbf{then} \label{adag:allreferenced}\\
        \> \> \> \( \var{DAG[v.round]} \gets \var{DAG[v.round]} \cup \{\var{v}\} \) \label{adag:add-new-vertex}\\
        \> \> \> \( \var{buffer} \gets \var{buffer} \setminus \{\var{v}\} \) \label{}\\

        \> \textbf{if} \( \exists Q \in \CQ_i : Q \models \var{DAG[r]} \) \textbf{then} \label{line:receive_quorum}\\
        \> \> \textbf{if} \(\var{r} \mod 4 = 0\) \textbf{then} \label{line:round4}\\
        \> \> \> \(\op{waveReady(\var{r}/4)}\) \label{}\\
        \> \> \> \(\var{r} \gets \var{r} + 1\) \label{}\\
        \> \> \> \(\var{v} \gets \op{createNewVertex(\var{r})} \) \label{line:newv1}\\
        \> \> \> \(\op{arb-broadcast((\var{v}, \var{r}))} \) \label{}\\
        \> \> \textbf{if} \(\var{r} \mod 4 = 1\) \textbf{then} \label{line:round1}\\
        \> \> \> \(\var{r} \gets \var{r} + 1\) \label{}\\
        \> \> \> \(\var{v} \gets \op{createNewVertex(\var{r})} \) \label{line:newv2}\\
        \> \> \> \(\op{arb-broadcast((\var{v}, \var{r}))} \) \label{}\\
        \> \> \textbf{if} \(\var{r} \mod 4 = 2 \wedge \var{tReady}\) \textbf{then} \label{line:round2}\\
        \> \> \> \(\var{r} \gets \var{r} + 1\) \label{}\\
        \> \> \> \(\var{v} \gets \op{createNewVertex(\var{r})} \) \label{line:newv3}\\
        \> \> \> \(\var{tReady} \gets \false\) \label{}\\
        \> \> \> \(\var{ack} \gets [\bot]^N\) \label{}\\
        \> \> \> \(\var{ready} \gets [\bot]^N\) \label{}\\
        \> \> \> \(\var{confirm} \gets [\bot]^N\) \label{}\\
        \> \> \> \(\op{arb-broadcast((\var{v}, \var{r}))} \) \label{}\\
        \> \> \textbf{if} \(\var{r} \mod 4 = 3\) \textbf{then} \label{line:round3}\\
        \> \> \> \(\var{r} \gets \var{r} + 1\) \label{}\\
        \> \> \> \(\var{v} \gets \op{createNewVertex(\var{r})} \) \label{line:newv4}\\
        \> \> \> \(\op{arb-broadcast((\var{v}, \var{r}))} \) \label{}\\

         \end{numbertabbing}

  }
  \caption{Asymmetric DAG-based consensus, part 1 (proc.~$p_i$). Algorithms~\ref{alg:adag2}~and~\ref{alg:adag3} contain the rest of the protocol. }
  \label{alg:adag}
\end{algo}

\begin{algo}
  \vbox{
  \small
  \begin{numbertabbing}
    xx\=xx\=xx\=xx\=xx\=xx\=MMMMMMMMMMMMMMMMMMM\=\kill

    \textbf{upon}  receiving $\msg{Ack}{p_j}$ from $p_j$ \textbf{do} \label{}\\  
        \> $\op{\var{ack[j]}} \gets \true$ \label{}\\
    \\

    \textbf{upon} $\exists Q \in \CQ_i : Q \models ack$ \textbf{do} \label{line:ack_quorum}\\  
        \> send $\msg{Ready}{}$ to all \label{}\\
    \\

    \textbf{upon}  receiving $\msg{Ready}{}$ from $p_j$ \textbf{do} \label{}\\  
        \> $\var{ready[j]} \gets \true$ \label{}\\
    \\

    \textbf{upon} $\exists Q \in \CQ_i : Q \models \var{ready}$ \textbf{do} \label{line:ready_quorum}\\  
        \> send $\msg{Confirm}{}$ to all \label{}\\
    \\

    \textbf{upon}  receiving $\msg{Confirm}{}$ from $p_j$ \textbf{do} \label{}\\  
        \> $\var{confirm[j]} \gets \true$ \label{}\\
    \\

    \textbf{upon} $\exists K \in \CK_i : K \models \var{confirm}$ \textbf{do} \label{}\\  
        \> send $\msg{Confirm}{}$ to all \label{}\\
    \\

    \textbf{upon} $\exists Q \in \CQ_i : Q \models \var{confirm}$ \textbf{do} \label{}\\  
        \> send $\msg{Confirm}{}$ to all \label{}\\
    \\

    \textbf{upon}  receiving $\msg{Confirm}{}$ from a quorum $Q \in \CQ_i$ \textbf{do} \label{line:ready2_quorum}\\  
        \> $\var{tReady} \gets \true$ \label{}\\

         \end{numbertabbing}

  }
  \caption{Asymmetric DAG-based consensus, part 2 (proc.~$p_i$)}
  \label{alg:adag2}
\end{algo}

\begin{algo}
  \vbox{
  \small
  \begin{numbertabbing}
    xx\=xx\=xx\=xx\=xx\=xx\=MMMMMMMMMMMMMMMMMMM\=\kill

    \textbf{upon}  \op{arb-deliver((\var{v}, \var{round}, $p_j$))  \textbf{do}  // asymmetric reliable broadcast delivery interface} \label{adag:deliver-vertex}\\  
        \> \(\var{v.source} \gets p_j\) \label{adag:v.source}\\
        \> \(\var{v.round} \gets \var{round}\) \label{adag:v.round}\\ 
        \> \textbf{if} \( \exists Q \in \CQ_j \text{ for some }\CQ_j \in \mathbb{Q} : Q \models  \var{v.strongEdges} \) \textbf{then} \label{adag:q-previous-round}\\
        \> \> \(\var{buffer} \gets \var{buffer} \cup \{\var{v}\}\) \label{}\\
        \> \textbf{if} \(\var{round} \mod 4 = 2\) \textbf{then} \label{}\\
        \> \> send $\msg{Ack}{p_i}$ to $p_j$ \label{line:ack_vertex}\\
    \\

     \textbf{upon}  \op{aa-broadcast(\var{b})} \textbf{do} // asymmetric atomic broadcast interface  \label{}\\  
        \> \(\var{blocksToPropose.enqueue(b)}\) \label{}\\
    \\

    \textbf{upon}  \op{waveReady(\var{w})}  \textbf{do} \label{}\\  
        \> \(\var{v} \gets \op{getWaveVertexLeader(\var{w})}\) \label{line:getleader}\\ 
        \> \textbf{if} \( (\var{v} = \bot) \vee (\nexists Q \in \CQ_j \text{ for some }\CQ_j \in \mathbb{Q}, \forall v' \in Q: \op{strongPath(\var{v}, 
        \var{v}')}) \) \textbf{then}  \label{line:commitrule}\\
        \> \> \textbf{return} \label{}\\
        \> \(\var{leadersStack.push(v)}\) \label{}\\
        \> \textbf{for} \( w' =(w-1),\dots,\var{decidedWave}+1\) \textbf{do} \label{adag:committed-in-order}\\
        \> \> \( \var{v}' \gets \op{getWaveVertexLeader(\var{w}')}\) \label{}\\
        \> \> \textbf{if } \( v' \neq \bot \wedge \op{strongPath(\var{v}, \var{v}')} \) \textbf{then} \label{}\\
        \> \> \> \( \var{leadersStack.push(v$'$)}\) \label{}\\
        \> \> \> \( \var{v} \gets \var{v}'\) \label{}\\
        \> \(\var{decidedWave} \gets w\) \label{}\\
        \> \(\op{orderVertices(\var{leaderStack})}\) \label{}\\
    \\

    \textbf{procedure}  \op{getWaveVertexLeader(\var{w})}   \label{}\\  
        \> \(p_j \gets \op{chooseLeader(\var{w})}\) \label{}\\
        \> \textbf{if } \( \exists v \in \var{DAG[\op{waveToRound(\var{w}, 1)}]} \text{ : } \var{v.source}=p_j \)  \textbf{then} \label{}\\
        \> \> \textbf{return }\(\var{v}\) \label{}\\
        \> \textbf{return }\(\bot\) \label{}\\
    \\

    \textbf{procedure}  \op{orderVertices(\var{leadersStack}, \var{w})} \label{adag:order-vertices}\\  
        \> \textbf{while} \(\neg \var{leadersStack.isEmpty}()\) \textbf{do} \label{}\\  
        \> \> \( \var{v} \gets \var{leadersStack.pop()}\) \label{}\\  
        \> \>  \(\var{verticesToDeliver} \gets \{ \var{v}' \in \cup_{r>0} \var{DAG[r]} \text{ s.t. } \op{path(\var{v}, \var{v}')} \wedge \var{v}' \notin \var{deliveredVertices}  \}\)  \label{}\\  
        \> \> \textbf{for}   \( v' \in \var{verticesToDeliver}\)  in some deterministic order \textbf{do}\label{}\\  
        \> \> \>  \textbf{output} \( \op{aa-deliver(\var{v}'\var{.block})} \) \label{adag:deliver-block}\\  
        \> \> \>  \(\var{deliveredVertices} \gets \var{deliveredVertices} \cup \{\var{v}'\}\) \label{}\\  

         \end{numbertabbing}

  }
  \caption{Asymmetric DAG-based consensus, part 3 (proc.~$p_i$)}
  \label{alg:adag3}
\end{algo}

\subsection{Proof of Correctness of Algorithms~\ref{alg:adag},~~\ref{alg:adag2}~and~~\ref{alg:adag3}}
\label{ap:adag-proof}

In Lemma~\ref{lem:adag-leaderpath} we prove that if a process in the maximal guild commits a leader vertex $v$ then eventually all other processes in the maximal guild will be able to commit this leader as well. This is done by proving that there will exist a strong path from any future leader vertex $u$ towards $v$. Since a vertex is committed if it is reachable from a leader, this proves that all other processes will commit $v$ eventually. When a future leader $u$ is chosen and committed, a graph traversal begins that includes all reachable vertices, which, according to Lemma \ref{lem:adag-leaderpath}, would include $v$.
\begin{lemma}
    \label{lem:adag-leaderpath}
    If some process in the maximal guild $p_i$ commits the leader vertex $v$ of a wave $w$, then for every leader vertex $u$ of a wave $w'>w$ and for every process $p_j$, if $u \in DAG_j[round(w', 1)]$, then \op{strong\_path(u, v)} returns true in wave $w'$
\end{lemma}
\begin{proof}

     For a process in the maximal guild $p_x$ and round $t$ it holds that any vertex (generated by an arbitrary process $p_y$) contained in any $DAG_x[t]$ has a strong path to the vertices produced by some quorum $Q' \in \CQ_y$ in $DAG_x[t-1]$. This follows from the fact that a process does not proceed to the next round until it has received enough messages for the current round (line~\ref{line:receive_quorum}). Therefore, any vertex in $DAG_x[t]$ will have a strong path to the vertices produced by some quorum $Q'' \in \CQ_z$ for some arbitrary process $p_z$ in any round $t''<t$.

    Given that process $p_i$ committed vertex $v$ there exists a quorum $Q_1 \in \CQ_i$ such that each $p \in Q_1$  produced a vertex  $v' \in DAG_i[\op{round(w, 4)}]$ and \op{strong\_path(v, v')} is true. 
    
    Consider any leader vertex $u$ of a future wave $w'>w$ and an arbitrary process in the maximal guild $p_j$. Denote by $Q_2$ the quorum such that $u$ has strong paths to the vertices produced by each $p'\in Q_2$ in \op{round(w, 4)}. By the quorum consistency property we know that there exists at least one process $p \in Q_1\cap Q_2$. Therefore, vertex $u$ has a strong path to the vertex of process $p'$ in \op{round(w, 4)} and this vertex has a strong path to vertex $v$ in \op{round(w, 1)}. This shows that any leader will have a strong path to any previously committed leader vertex.
\end{proof}

The algorithm makes progress only if the leader is committed by at least one wise process in the maximal guild. If this never happens, then no messages are delivered and liveness is lost. Therefore we need a way to calculate, in expectation, how many waves will go by before a leader is committed by at least one process. Recall that a leader, which is a vertex in the first round of the wave, is considered committed by a process if it has strong paths from the vertices created by one of its quorums in the last round of the wave. Lemma~\ref{lem:adag-bound1} will help us calculate the probability by identifying a minimum bound on the size of a set of vertices from the first round such that the commit rule holds for all of them. 

\begin{lemma}
    \label{lem:adag-bound1}
    Let $p_i$ be a process in the maximal guild that completes wave $w$. Then there are quorums $Q\in \CQ_m$ for some process $p_m$ in the maximal guild, $Q' \in \CQ_i$, a set $U \subseteq DAG_i[\op{round(w, 1)}]$ and a set $V \subseteq DAG_i[\op{round(w, 4)}]$ such that $U$ contains the vertices created by the processes in $Q$, $V$ contains the vertices created by the processes in $Q'$ and for all $ u \in U, v \in V$ \op{strong\_path(u, v)} is true. 
\end{lemma}
\begin{proof}

    Just as in the proof of Lemma 2 in the DAG-Rider protocol~\cite{DBLP:conf/podc/KeidarKNS21} we will show that each wave executes an asymmetric gather. This will show that there is a set of processes of bounded minimum size in the first round such that, if any of them is chosen as the leader, the algorithm will commit.

    The gather wave structure is set in the \emph{while} loop in Line~\ref{line:adag-while}. This section manages the change from round to round and from wave to wave, dictated by the value of the round counter $r$. We proceed to show how one wave of the algorithm follows the structure of the asymmetric gather Algorithm~\ref{alg:core}. 

    When $r\ \text{mod}\ 4$ is 1 then the algorithm behaves as during the first round of the gather. That is, it waits until $Q \models \op{DAG}[r]$ and then broadcasts a new vertex that references the vertices in $\op{DAG}[r]$. This is equivalent to sending the $\str{DistributeS}$ in the gather. 

    When a process has received acknowledgments back from one of its quorums (Line~\ref{line:ack_quorum}), it sends a $\str{Ready}$ message. This follows the same idea of the asymmetric gather. Upon receiving $\str{Ready}$ messages from a quorum (Line~\ref{line:ready_quorum}), it sends a $\str{Confirm}$ message. Once a process has received $\str{Confirm}$ messages from one of its quorums, it sets the value of variable $tReady$ to true. This allows it to continue processing vertices for the second round of the wave and is equivalent to the \emph{waiting} added to the asymmetric gather before sending the $T$ sets. 

    Once \op{tReady} is true, $r\ \text{mod}\ 4$ is 2, and the process has received vertices from one of its quorums it proceeds to create and broadcast a new vertex. This vertex references all vertices in DAG[r] and is equivalent to sending the $T$ sets in Algorithm~\ref{alg:core}. 

    When $r\ \text{mod}\ 4$ is 3 and the process has received vertices from one of its quorums, it creates and broadcasts a new vertex. This vertex is the equivalent of the $U$ set in Algorithm~\ref{alg:core}. As proved there, this set contains a common core. All other processes that have executed the algorithm correctly will have a common core in the causal history of their vertex for the current round. Therefore, once $r\ \text{mod}\ 4$ is 0 and the process has received vertices for this round coming from one of its quorums, we know that there exists a quorum $Q\in \CQ_m$ for some wise process $p_m$ such that all vertices belonging to the last round of the wave will contain strong paths to all vertices created by processes in $Q$ during the first round. This quorum $Q$ is the common core guaranteed by the gather algorithm. During a round each process waits to receive messages from at least one of their quorums, therefore we know that $Q' \models DAG_i[round(w, 4)]$ for some $Q' \in \CQ_i$. Since all these vertices have strong paths to the vertices created by processes in $Q$ in the first round we know that for all $ u \in U, v \in V$ \op{strong\_path(u, v)} is true. 
    \end{proof}

Now that we know there is a common core of processes for which the commit rule holds true, we can use this information to calculate the expected number of waves until a participant commits a leader. In the next lemma, we will use the notation $c(\mathbb{Q})$ to determine $\min_{Q\ \in \CQ_i, \forall p_i \in \mathcal{P}} |Q|$, i.e., the size of the smallest quorum contained in $\mathbb{Q}$.
\begin{lemma}
    The expected number of waves until the commit rule is met is at most $\frac{|\mathcal{P}|}{c(\mathbb{Q})}$.
    \label{lem:rider-expected}
\end{lemma}

\begin{proof}
    The probability that one process commits a leader during a wave is lower bounded by the probability that the leader is a part of the common core quorum $Q$ for that wave. Therefore we obtain that the probability of a leader being committed is $\geq P(\text{leader} \in Q) \geq \frac{|Q|}{|\mathcal{P}|}\geq \frac{c(\mathbb{Q})}{|\mathcal{P}|}$. 

    The number of waves until the commit rule is met is geometrically distributed with success probability $\frac{c(\mathbb{Q})}{|\mathcal{P}|}$. Thus, the expected number of waves is bounded by $\frac{|\mathcal{P}|}{c(\mathbb{Q})}$. 
\end{proof}

This shows that, in expectation, each member of the maximal guild will commit messages frequently, with the time gaps between commits being of constant duration. The duration of the time gaps depends on the quorum system of each participant. 

The next three lemmas will prove simple properties of Algorithm \ref{alg:adag} that will be useful to prove the validity and agreement  properties of the algorithm. 

\begin{lemma}
    In any execution with a guild, if a process $p_i$ in the maximal guild adds a vertex $v$ to its $\op{DAG}_i$, then eventually all processes in the maximal guild will add $v$ to their $\op{DAG}$
    \label{lem:rider-allverticesadded}
\end{lemma}

\begin{proof}
    We will prove this by induction on the number of rounds. The base case is round 0. If a process $p_i$ in the maximal guild delivers a vertex and adds it to its $DAG_i$ (line \ref{adag:add-new-vertex}) this means that it passed all the conditions to be added, namely
    \begin{enumerate}
        \item It was delivered by the reliable broadcast (line \ref{adag:deliver-vertex})
        \item It contains references to a quorum $Q$ on a previous round (line \ref{adag:q-previous-round})
        \item It contains the correct values in $v.source$ and $v.round$ (lines \ref{adag:v.source} and \ref{adag:v.round})
        \item All the vertices referenced by $v$ must be in $DAG_i$ (line \ref{adag:allreferenced})
    \end{enumerate}

    From condition number 1 and the totality property of the asymmetric reliable broadcast we know that all other processes in the maximal guild will also deliver $v$. By conditions number 2 and 3 we know that these checks will also pass in the other processes in the maximal guild since all of them are local checks. Condition number 4 is satisfied since we are currently in round 0 and therefore $v$ cannot reference vertices in previous rounds. 

    We will assume that this holds for all vertices delivered in any round $r'<r$. We will then show that it also applies in round $r$.

    Let us suppose a process in the maximal guild $p_i$ has delivered and added to its $\op{DAG}_i$ a vertex $v$. We show that all other processes in the maximal guild will eventually add $v$ to their DAG. By the totality condition of the reliable broadcast, we obtain that all other members of the maximal guild will also eventually deliver it. We also know that conditions number 2 and 3 will be satisfied since these are local checks done at every process. Finally, we must show that all other processes in the maximal guild contain all the vertices referenced by $v$. We know that these vertices belong to rounds previous to $r$ and that they are contained in the DAG of process $p_i$. By the induction hypothesis, they will also be added to the DAG of all other processes in the maximal guild and therefore all of these processes will contain all vertices referenced by $v$. This proves that $v$ will also be added to the DAG of all other processes in the maximal guild. 
\end{proof}

\begin{lemma}
    In any execution with a guild, every vertex that is broadcast by a correct process is eventually added to the DAG of all processes in the maximal guild
    \label{lem:rider-broadcastedadded}
\end{lemma}
\begin{proof}
    When a correct process $p_i$ broadcasts a vertex $v$, it broadcasts a valid vertex. By the validity property of the asymmetric reliable broadcast, $v$ will be eventually delivered to some process in the maximal guild $p_j$ and added to its own $DAG_j$ since it is a valid vertex. By Lemma \ref{lem:rider-allverticesadded} this implies that all other processes in the maximal guild will add $v$ to their DAGs. 
\end{proof}

\begin{lemma}
    If a process in the maximal guild $p_i$ adds a vertex $v$ to its DAG then all of its causal history is already in its DAG
    \label{lem:rider-causalhistory}
\end{lemma}
\begin{proof}
    We will prove this by induction on the number of rounds. For the base case $r=0$ we know that a vertex $v$ is added to $DAG_i$ only if all of the vertices it references are already in $DAG_i$ (line~\ref{adag:allreferenced}). Since we are in the first round this means that $v$ cannot reference vertices in previous rounds and therefore it holds that all of its causal history is already in $DAG_i$ and therefore $v$ is added to the DAG (line \ref{adag:add-new-vertex}).

    We assume this holds for all rounds $r'<r$. We now prove that it will also be true in round $r$. Lets suppose process $p_i$ delivers a vertex $v$ and adds it to its $DAG_i$. Line~\ref{adag:allreferenced} guarantees that $v$ is added only if the vertices towards which it has links are already in $DAG_i$. By the induction hypothesis, we know that all these vertices already have all their causal history in $DAG_i$. Therefore all of $v$'s causal history is also in $DAG_i$.
\end{proof}

The following lemma proves that Algorithm \ref{alg:adag} satisfies the first property of Byzantine atomic broadcast: total order. 
\begin{lemma}
    Algorithm \ref{alg:adag} satisfies the total order property of the asymmetric Byzantine atomic broadcast problem.
    \label{lem:adag-to}
\end{lemma}
\begin{proof}
    We start by proving that for all processes in the maximal guild each wave $w$ can have at most one leader $v$. For $v$ to be the wave leader it must be returned by the common coin. The matching property of the asymmetric common coin guarantees that all processes in the maximal guild will obtain the same coin value and by the agreement property of the asymmetric reliable broadcast. We also know that vertices are committed in an increasing wave number (line \ref{adag:committed-in-order}). If a process in the maximal guild commits some vertex $v$, Lemma \ref{lem:adag-leaderpath} guarantees that there will be paths towards it from any vertex $u$ in future waves that might be committed. 

    By combining these claims we know that if two processes in the maximal guild commit the same wave leader, they do so in the same order. 

    After committing a wave vertex leader $v$ a process $p_i$ atomically delivers all of $v$'s causal history in some deterministic order that is identical for all processes (line \ref{adag:order-vertices}). We also know that a vertex is added only after all of its causal history is already in the DAG as per Lemma \ref{lem:rider-causalhistory}. Since all processes in the maximal guild commit the same leader, in the same order, and with the same causal histories we can conclude that all processes in the maximal guild deliver the blocks in the same order. 
\end{proof}

In the next two lemmas we prove that Algorithm \ref{alg:adag} satisfies the missing properties to obtain an asymmetric Byzantine atomic broadcast.

\begin{lemma}
    Algorithm \ref{alg:adag}  satisfies the agreement property of the asymmetric Byzantine atomic broadcast problem
    \label{lem:adag-agreement}
\end{lemma}
\begin{proof}
    If a process in the maximum guild outputs $a\_deliver(b, r, p_k)$ (line \ref{adag:deliver-block}) it means that $b$ is a block in some vertex $u$. This vertex $u$ was in the causal history of the leader $v$ of some wave $w$. According to Lemma \ref{lem:rider-expected}, every process in the maximal guild that has not yet committed vertex $v$ will eventually reach a wave $w'>w$ in which the commit rule is satisfied. By Lemma \ref{lem:adag-leaderpath} we know that the leader of wave $w'$ will have a path towards $v$ and therefore will process its uncommitted causal history, including block $b$. 
\end{proof}

\begin{lemma}
    Algorithm \ref{alg:adag} satisfies the validity property of the asymmetric Byzantine agreement problem
    \label{lem:adag-validity}
\end{lemma}
\begin{proof}
    Lemma \ref{lem:rider-broadcastedadded} guarantees that any vertex $v$ broadcasted by a correct process will be added to the DAG of all processes in the maximal guild. When a vertex is broadcasted the algorithm (procedure \op{set\_weak\_edges} and line \ref{adag:set-strong-edges}) guarantees  that it has a path to all vertices in previous rounds. Therefore a future wave leader will have a path towards $v$, which guarantees that it will be delivered. 
\end{proof}

Based on Lemmas \ref{lem:adag-agreement}, \ref{lem:adag-to} and \ref{lem:adag-validity} we can conclude that Algorithm \ref{alg:adag} meets the criteria for an asymmetric Byzantine atomic broadcast algorithm. Thus, we have proven that it solves the asymmetric Byzantine atomic broadcast problem.

\subsection{Relation to Other DAG-Based Protocols}

Algorithm~\ref{alg:adag} was developed as an asymmetric translation of the DAG-Rider protocol~\cite{DBLP:conf/podc/KeidarKNS21}. In this sense, we take its three main components (Reliable Broadcast, Common Coin, and Gather) and replace them with asymmetric versions. The asymmetric gather is introduced in this work after showing that existing approaches are unsuitable for the asymmetric setting. We show in Section~\ref{sec:aGather} the impossibility of maintaining the elegant 4-round structure of the symmetric gather used by DAG-Rider. Instead, we introduce a modified asymmetric gather protocol that overcomes the challenges posed by the usage of asymmetric quorums. This is of direct relevance for obtaining asymmetric versions of other DAG-based protocols that also use a common core. 

Many DAG-based consensus algorithms~\cite{DBLP:conf/podc/KeidarKNS21, DBLP:conf/eurosys/DanezisKSS22, DBLP:conf/ccs/SpiegelmanGSK22, DBLP:conf/wdag/KeidarNPS23} use commit rules that rely on the existence of a common core. In this sense, our work addresses the crucial problem of obtaining the common core in a constant number of rounds in the asymmetric setting. This is of importance to maintain a low latency. We also show how, due to the properties of asymmetric quorums, maintaining the elegant symmetric approach (like in Algorithm~\ref{alg:asymgather-try}) would prevent us from having a constant latency. The common core tools developed here do not apply to protocols such as Mysticeti~\cite{DBLP:conf/ndss/BabelCDKKKST25}, which forego these techniques to further reduce latency. 

The communication protocol used to construct the DAG is another factor that differentiates DAG-Rider from other DAG-based protocols. Reliable Broadcast requires multiple round-trips to broadcast each block between processes, which leads to an increase in the latency of DAG-Rider. Protocols like Mysticeti or Cordial Miners abandon the idea of having a reliable channel, while Tusk~\cite{DBLP:conf/eurosys/DanezisKSS22}, Shoal~\cite{DBLP:journals/corr/abs-2306-03058}, and Bullshark ~\cite{DBLP:conf/ccs/SpiegelmanGSK22} replace reliable broadcast with Narwhal~\cite{DBLP:conf/eurosys/DanezisKSS22}. Relevant asymmetric equivalents would have to be developed in order to translate the aforementioned protocols to the asymmetric model. In particular, future work could address the problem of redesigning Narwhal for the asymmetric setting. Given the separation between the mempool and consensus layer employed by Tusk and Bullshark, this would also help to develop asymmetric versions of these two. In addition, this would lead to the development of a practical asymmetric DAG-based consensus. Asymmetric DAG-Rider (like its symmetric counterpart) requires unbounded memory in order to provide fairness, which makes it unfit for a practical system.

\section{Conclusion}
\label{sec:conclusion}
This work introduced the first DAG-based consensus protocol in the asymmetric-trust model, exploring the unaddressed relationship between asymmetric quorums and DAG-based consensus. Specifically, we built upon and generalized DAG-Rider. In particular, we have shown that standard techniques commonly used in this domain, such as common core primitives, were unsuitable for the asymmetric setting and required modifications to ensure a sound protocol. To this end, we have introduced the first asymmetric gather protocol, which operates in a constant number of rounds. Several DAG-based consensus protocols have emerged in recent years, some of which are used in real-world systems like blockchains. We believe the techniques introduced in this paper can also be applied to other DAG-based consensus protocols to adapt them to the asymmetric setting.

\begin{acks}
This work has been funded by the Swiss National Science Foundation (SNSF)
under grant agreement Nr\@.~219403 (Emerging Consensus), by the Initiative for
Cryptocurrencies and Contracts (IC3), and by the Cryptographic Foundations for
Digital Society, CryptoDigi, DFF Research Project 2, Grant ID
10.46540/3103-00077B.
\end{acks}

\clearpage

\appendix

\section{Proof of Lemma~\ref{lem:agather-no-common-core}}
\label{ap:gather-counterexample}
In this section we prove Lemma~\ref{lem:agather-no-common-core}, which states that Algorithm~\ref{alg:asymgather-try} does not satisfy the common core primitive. Using a counterexample quorum system, we show that the straightforward asymmetric gather implementation based on the quorum-replacement idea of Alpos\etal~\cite{DBLP:journals/dc/AlposCTZ24} does not achieve a common core after three rounds of communication. For this, we show that after the execution of the algorithm, there exists no common core within the output sets of each participant. We use the quorum system defined in Figure~\ref{fig:ce-failprone}, which contains only one quorum for each process. In this execution we will assume that all processes are correct, therefore wise, and that there are no failures during the execution. As all processes are wise, the maximal guild is composed by all the 30 processes. 

As mentioned in Section~\ref{sec:aGather}, a common core $S^+$ is a set composed of the initial values proposed by the members of a quorum $Q$ for some process $p_i$. In this counterexample we will execute the algorithm to obtain the output sets of each process and show that there is not one single set $S^+$ that is contained in the outputs of all processes. 

During each round, every process will wait to receive messages from one of its quorums, which it will then combine into a new set that will be distributed during the next round. We make use of three different figures to guide the reader through the counterexample and we use the terminology from Algorithm~\ref{alg:asymgather-try}. 

Figure~\ref{fig:counterexample-1} shows the $S$ sets of each process formed by the received values during the first round. Each row represents a process, and the blue columns show the processes from which sets were received. All the received messages will be combined into a set that will be distributed during the next round.

Figure~\ref{fig:counterexample-2} shows all the processes from which each process has received messages after the second communication round, also denoted $T$. That is, the union of the $S$ sets sent by the members of the quorum of each process. For example, process 1 obtains its $T$ set as the union of the $S$ sets of processes 1, 2, 3, 4, 5, and 16, the members of its only quorum. All this values are combined in a set denoted $T$, which is distributed during the next and final round. 

Figure~\ref{fig:counterexample-3} shows all the processes from which each process has received messages after receiving $T$ sets from its quorum. The union of all these values form the $U$ sets. These are the sets delivered by each process at the end of the execution of the algorithm. Note that there is no common core among the $U$ sets of each process. That is, there is no single $S$ set that is contained in all $U$ sets. An easy way to observe this is to note that all $S$ sets contain at least one process in the range $[16, 30]$ while the $U$ sets portrayed in Figure~\ref{fig:counterexample-3} are all missing at least one message from a process in the aforementioned range. This means that no $S$ set is contained in the output sets of all processes. 

It is possible to prove that, due to the consistency property of asymmetric quorums, the common core property can be obtained in $\log n -1$ rounds of the type of communication used here. That is, after a logarithmic number of rounds of receiving messages from a quorum and sending messages, there will be a common core in the $U$ sets. This gives an asymmetric gather with a higher execution time, which is not adequate for our purposes. This motivates the development of Algorithm~\ref{alg:core}, a constant-round asymmetric gather.

In Listing~\ref{lst:python-counterexample} we provide a Python script to verify the results. It consists of three rounds of creating and distributing sets. During the first round, each process creates the $S$ set by merging the initial values from the processes in its quorum. 

During the second and third rounds the processes create the $T$ and $U$ sets in the same way, merging the values in the $S$ and $T$ sets, respectively, from the processes in its quorum. 

In the end, the $U$ sets of all processes are checked to see if there exists some $S$ set that is present in the output sets of all processes. The variable \emph{quorums\_per\_process} will contain the $S$ sets contained in the $U$ sets of each process. In the final \emph{for} loop we check if there is some $S$ set that is contained in the $U$ sets of all processes. The variable \emph{all\_candidates} will contain the $S$ sets contained in all final $U$ sets. As there are none, it will be empty.

\clearpage
\begin{lstlisting}[language=Python, caption=Python script to verify the absence of a common core in Algorithm~\ref{alg:asymgather-try}, label={lst:python-counterexample}, basicstyle=\ttfamily\small, breaklines=true, keywordstyle=\color{blue}, commentstyle=\color{purple}, stringstyle=\color{red}]
### counterexample with 30 processes
quorums = {}
#quorums of each process, as in figure 1
quorums[1] = {1, 2, 3, 4, 5, 16}
quorums[2] = {1, 6, 7, 8, 9, 17}
quorums[3] = {1, 2, 3, 4, 5, 18}
quorums[4] = {1, 6, 7, 8, 9, 19}
quorums[5] = {2, 6, 10, 11, 12, 20}
quorums[6] = {4, 8, 11, 13, 15, 21}
quorums[7] = {4, 8, 11, 13, 15, 22}
quorums[8] = {5, 9, 12, 14, 15, 23}
quorums[9] = {5, 9, 12, 14, 15, 24}
quorums[10] = {4, 8, 11, 13, 15, 25}
quorums[11] = {1, 6, 7, 8, 9, 26}
quorums[12] = {2, 6, 10, 11, 12, 27}
quorums[13] = {3, 7, 10, 13, 14, 28}
quorums[14] = {3, 7, 10, 13, 14, 29}
quorums[15] = {5, 9, 12, 14, 15, 30}
quorums[16] = {1, 2, 3, 4, 5, 16}
quorums[17] = {1, 2, 3, 4, 5, 16}
quorums[18] = {1, 2, 3, 4, 5, 16}
quorums[19] = {1, 2, 3, 4, 5, 16}
quorums[20] = {1, 6, 7, 8, 9, 27}
quorums[21] = {1, 6, 7, 8, 9, 27}
quorums[22] = {1, 6, 7, 8, 9, 20}
quorums[23] = {2, 6, 10, 11, 12, 30}
quorums[24] = {2, 6, 10, 11, 12, 30}
quorums[25] = {1, 6, 7, 8, 9, 22}
quorums[26] = {1, 2, 3, 4, 5, 16}
quorums[27] = {1, 6, 7, 8, 9, 27}
quorums[28] = {1, 2, 3, 4, 5, 16}
quorums[29] = {1, 2, 3, 4, 5, 29}
quorums[30] = {2, 6, 10, 11, 12, 30}


# computing the first round sets: the s sets
s_sets = {}
for i in quorums:
    seen = []
    for j in quorums[i]:
        seen.append(j)
    s_sets[i] = set(seen)

# computing the second round sets: the t sets
t_sets = {}
for i in quorums:
    seen = set()
    for j in quorums[i]:
        seen = seen.union(s_sets[j])
    t_sets[i] = seen

# computing the third round sets: the sets u
u_sets = {}
for i in quorums:
    seen = set()
    for j in quorums[i]:
        seen = seen.union(t_sets[j])
    u_sets[i] = seen


s_sets_per_process = {}

# for the u set of each process, check which s sets it contains
for i in u_sets:
    contains = []
    # check all s sets
    for j in s_sets:
        if(s_sets[j].issubset(u_sets[i])):
            contains.append(j)
    s_sets_per_process[i] = set(contains)

all_candidates = set([i for i in range(1, 31)])

# check if there is some s set in the u set of all processes
for i in s_sets_per_process:
    all_candidates = all_candidates.intersection(s_sets_per_process[i])
# if all_candidates = set() then there is no s set contained 
# in the u set of all processes
print(all_candidates)

\end{lstlisting}

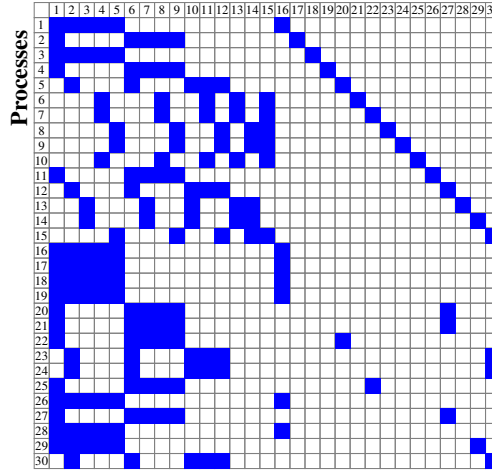
\begin{figure}[ht]
\centering
\begin{tikzpicture}[scale=0.2]
    \def\rows{30}
    \def\cols{30}

    \foreach \x in {0,...,\cols} {
        \foreach \y in {0,...,\rows} {
            \draw[gray] (\x, \y) rectangle (\x+1, \y+1);
        }
    }

    \foreach \x in {1,...,30} {
        \pgfmathsetmacro{\result}{int(31- \x)} 
        \node[align=center, outer sep=0pt] at (\x + 0.5, 30.5) {\fontsize{6}{0}\selectfont \textcolor{black}{\scalebox{.8}{\x}}};
    }
    \foreach \y in {1,...,30} {
        \pgfmathsetmacro{\result}{int(31- \y)} 
        \node[align=center, text width=1cm] at (0.5, \y - 0.5) {\fontsize{6}{0}\selectfont \textcolor{black}{\scalebox{.8}{\result}}};
    }

\def\yy{1}
\pgfmathsetmacro{\y}{int(30- \yy)}
\foreach \x in {16, 1, 2, 3, 4, 5} {
    \fill[\scolor, fill opacity=\sopacity] (\x, \y) rectangle (\x+1, \y+1);
}
\def\yy{2}
\pgfmathsetmacro{\y}{int(30- \yy)}
\foreach \x in {1, 17, 6, 7, 8, 9} {
    \fill[\scolor, fill opacity=\sopacity] (\x, \y) rectangle (\x+1, \y+1);
}
\def\yy{3}
\pgfmathsetmacro{\y}{int(30- \yy)}
\foreach \x in {1, 2, 3, 4, 5, 18} {
    \fill[\scolor, fill opacity=\sopacity] (\x, \y) rectangle (\x+1, \y+1);
}
\def\yy{4}
\pgfmathsetmacro{\y}{int(30- \yy)}
\foreach \x in {1, 19, 6, 7, 8, 9} {
    \fill[\scolor, fill opacity=\sopacity] (\x, \y) rectangle (\x+1, \y+1);
}
\def\yy{5}
\pgfmathsetmacro{\y}{int(30- \yy)}
\foreach \x in {2, 20, 6, 10, 11, 12} {
    \fill[\scolor, fill opacity=\sopacity] (\x, \y) rectangle (\x+1, \y+1);
}
\def\yy{6}
\pgfmathsetmacro{\y}{int(30- \yy)}
\foreach \x in {4, 21, 8, 11, 13, 15} {
    \fill[\scolor, fill opacity=\sopacity] (\x, \y) rectangle (\x+1, \y+1);
}
\def\yy{7}
\pgfmathsetmacro{\y}{int(30- \yy)}
\foreach \x in {4, 22, 8, 11, 13, 15} {
    \fill[\scolor, fill opacity=\sopacity] (\x, \y) rectangle (\x+1, \y+1);
}
\def\yy{8}
\pgfmathsetmacro{\y}{int(30- \yy)}
\foreach \x in {5, 23, 9, 12, 14, 15} {
    \fill[\scolor, fill opacity=\sopacity] (\x, \y) rectangle (\x+1, \y+1);
}
\def\yy{9}
\pgfmathsetmacro{\y}{int(30- \yy)}
\foreach \x in {5, 24, 9, 12, 14, 15} {
    \fill[\scolor, fill opacity=\sopacity] (\x, \y) rectangle (\x+1, \y+1);
}
\def\yy{10}
\pgfmathsetmacro{\y}{int(30- \yy)}
\foreach \x in {4, 8, 25, 11, 13, 15} {
    \fill[\scolor, fill opacity=\sopacity] (\x, \y) rectangle (\x+1, \y+1);
}
\def\yy{11}
\pgfmathsetmacro{\y}{int(30- \yy)}
\foreach \x in {1, 6, 7, 8, 9, 26} {
    \fill[\scolor, fill opacity=\sopacity] (\x, \y) rectangle (\x+1, \y+1);
}
\def\yy{12}
\pgfmathsetmacro{\y}{int(30- \yy)}
\foreach \x in {2, 6, 27, 10, 11, 12} {
    \fill[\scolor, fill opacity=\sopacity] (\x, \y) rectangle (\x+1, \y+1);
}
\def\yy{13}
\pgfmathsetmacro{\y}{int(30- \yy)}
\foreach \x in {3, 7, 10, 28, 13, 14} {
    \fill[\scolor, fill opacity=\sopacity] (\x, \y) rectangle (\x+1, \y+1);
}
\def\yy{14}
\pgfmathsetmacro{\y}{int(30- \yy)}
\foreach \x in {29, 3, 7, 10, 13, 14} {
    \fill[\scolor, fill opacity=\sopacity] (\x, \y) rectangle (\x+1, \y+1);
}
\def\yy{15}
\pgfmathsetmacro{\y}{int(30- \yy)}
\foreach \x in {5, 30, 9, 12, 14, 15} {
    \fill[\scolor, fill opacity=\sopacity] (\x, \y) rectangle (\x+1, \y+1);
}
\def\yy{16}
\pgfmathsetmacro{\y}{int(30- \yy)}
\foreach \x in {16, 1, 2, 3, 4, 5} {
    \fill[\scolor, fill opacity=\sopacity] (\x, \y) rectangle (\x+1, \y+1);
}
\def\yy{17}
\pgfmathsetmacro{\y}{int(30- \yy)}
\foreach \x in {16, 1, 2, 3, 4, 5} {
    \fill[\scolor, fill opacity=\sopacity] (\x, \y) rectangle (\x+1, \y+1);
}
\def\yy{18}
\pgfmathsetmacro{\y}{int(30- \yy)}
\foreach \x in {16, 1, 2, 3, 4, 5} {
    \fill[\scolor, fill opacity=\sopacity] (\x, \y) rectangle (\x+1, \y+1);
}
\def\yy{19}
\pgfmathsetmacro{\y}{int(30- \yy)}
\foreach \x in {16, 1, 2, 3, 4, 5} {
    \fill[\scolor, fill opacity=\sopacity] (\x, \y) rectangle (\x+1, \y+1);
}
\def\yy{20}
\pgfmathsetmacro{\y}{int(30- \yy)}
\foreach \x in {1, 6, 7, 8, 9, 27} {
    \fill[\scolor, fill opacity=\sopacity] (\x, \y) rectangle (\x+1, \y+1);
}
\def\yy{21}
\pgfmathsetmacro{\y}{int(30- \yy)}
\foreach \x in {1, 6, 7, 8, 9, 27} {
    \fill[\scolor, fill opacity=\sopacity] (\x, \y) rectangle (\x+1, \y+1);
}
\def\yy{22}
\pgfmathsetmacro{\y}{int(30- \yy)}
\foreach \x in {1, 20, 6, 7, 8, 9} {
    \fill[\scolor, fill opacity=\sopacity] (\x, \y) rectangle (\x+1, \y+1);
}
\def\yy{23}
\pgfmathsetmacro{\y}{int(30- \yy)}
\foreach \x in {2, 6, 10, 11, 12, 30} {
    \fill[\scolor, fill opacity=\sopacity] (\x, \y) rectangle (\x+1, \y+1);
}
\def\yy{24}
\pgfmathsetmacro{\y}{int(30- \yy)}
\foreach \x in {2, 6, 10, 11, 12, 30} {
    \fill[\scolor, fill opacity=\sopacity] (\x, \y) rectangle (\x+1, \y+1);
}
\def\yy{25}
\pgfmathsetmacro{\y}{int(30- \yy)}
\foreach \x in {1, 6, 7, 8, 9, 22} {
    \fill[\scolor, fill opacity=\sopacity] (\x, \y) rectangle (\x+1, \y+1);
}
\def\yy{26}
\pgfmathsetmacro{\y}{int(30- \yy)}
\foreach \x in {16, 1, 2, 3, 4, 5} {
    \fill[\scolor, fill opacity=\sopacity] (\x, \y) rectangle (\x+1, \y+1);
}
\def\yy{27}
\pgfmathsetmacro{\y}{int(30- \yy)}
\foreach \x in {1, 6, 7, 8, 9, 27} {
    \fill[\scolor, fill opacity=\sopacity] (\x, \y) rectangle (\x+1, \y+1);
}
\def\yy{28}
\pgfmathsetmacro{\y}{int(30- \yy)}
\foreach \x in {16, 1, 2, 3, 4, 5} {
    \fill[\scolor, fill opacity=\sopacity] (\x, \y) rectangle (\x+1, \y+1);
}
\def\yy{29}
\pgfmathsetmacro{\y}{int(30- \yy)}
\foreach \x in {1, 2, 3, 4, 5, 29} {
    \fill[\scolor, fill opacity=\sopacity] (\x, \y) rectangle (\x+1, \y+1);
}
\def\yy{30}
\pgfmathsetmacro{\y}{int(30- \yy)}
\foreach \x in {2, 6, 10, 11, 12, 30} {
    \fill[\scolor, fill opacity=\sopacity] (\x, \y) rectangle (\x+1, \y+1);
}

\node at (13, 32) {\textbf{Messages received after one round of receiving messages}}; %
\node[rotate=90] at (-1, 26) {\textbf{Processes}}; %

\end{tikzpicture}
    \caption{Values possessed by each process after one round of hearing messages coming from its quorums. Also denoted as $S$ sets }
    \label{fig:counterexample-1}
\end{figure}

\begin{figure}[ht]
\centering
\begin{tikzpicture}[scale=0.2]
    \def\rows{30}
    \def\cols{30}

    \foreach \x in {0,...,\cols} {
        \foreach \y in {0,...,\rows} {
            \draw[gray] (\x, \y) rectangle (\x+1, \y+1);
        }
    }

    \foreach \x in {1,...,30} {
        \pgfmathsetmacro{\result}{int(31- \x)} 
        \node[align=center, outer sep=0pt] at (\x + 0.5, 30.5) {\fontsize{6}{0}\selectfont \textcolor{black}{\scalebox{.8}{\x}}};
    }
    \foreach \y in {1,...,30} {
        \pgfmathsetmacro{\result}{int(31- \y)} 
        \node[align=center, text width=1cm] at (0.5, \y - 0.5) {\fontsize{6}{0}\selectfont \textcolor{black}{\scalebox{.8}{\result}}};
    }

\def\yy{1}
\pgfmathsetmacro{\y}{int(30- \yy)}
\foreach \x in {1, 2, 3, 4, 5, 6, 7, 8, 9, 10, 11, 12, 16, 17, 18, 19, 20} {
    \fill[\scolor, fill opacity=\sopacity] (\x, \y) rectangle (\x+1, \y+1);
}
\def\yy{2}
\pgfmathsetmacro{\y}{int(30- \yy)}
\foreach \x in {1, 2, 3, 4, 5, 8, 9, 11, 12, 13, 14, 15, 16, 21, 22, 23, 24} {
    \fill[\scolor, fill opacity=\sopacity] (\x, \y) rectangle (\x+1, \y+1);
}
\def\yy{3}
\pgfmathsetmacro{\y}{int(30- \yy)}
\foreach \x in {1, 2, 3, 4, 5, 6, 7, 8, 9, 10, 11, 12, 16, 17, 18, 19, 20} {
    \fill[\scolor, fill opacity=\sopacity] (\x, \y) rectangle (\x+1, \y+1);
}
\def\yy{4}
\pgfmathsetmacro{\y}{int(30- \yy)}
\foreach \x in {1, 2, 3, 4, 5, 8, 9, 11, 12, 13, 14, 15, 16, 21, 22, 23, 24} {
    \fill[\scolor, fill opacity=\sopacity] (\x, \y) rectangle (\x+1, \y+1);
}
\def\yy{5}
\pgfmathsetmacro{\y}{int(30- \yy)}
\foreach \x in {1, 2, 4, 6, 7, 8, 9, 10, 11, 12, 13, 15, 17, 21, 25, 26, 27} {
    \fill[\scolor, fill opacity=\sopacity] (\x, \y) rectangle (\x+1, \y+1);
}
\def\yy{6}
\pgfmathsetmacro{\y}{int(30- \yy)}
\foreach \x in {1, 3, 5, 6, 7, 8, 9, 10, 12, 13, 14, 15, 19, 23, 26, 27, 28, 30} {
    \fill[\scolor, fill opacity=\sopacity] (\x, \y) rectangle (\x+1, \y+1);
}
\def\yy{7}
\pgfmathsetmacro{\y}{int(30- \yy)}
\foreach \x in {1, 3, 5, 6, 7, 8, 9, 10, 12, 13, 14, 15, 19, 20, 23, 26, 28, 30} {
    \fill[\scolor, fill opacity=\sopacity] (\x, \y) rectangle (\x+1, \y+1);
}
\def\yy{8}
\pgfmathsetmacro{\y}{int(30- \yy)}
\foreach \x in {2, 3, 5, 6, 7, 9, 10, 11, 12, 13, 14, 15, 20, 24, 27, 29, 30} {
    \fill[\scolor, fill opacity=\sopacity] (\x, \y) rectangle (\x+1, \y+1);
}
\def\yy{9}
\pgfmathsetmacro{\y}{int(30- \yy)}
\foreach \x in {2, 3, 5, 6, 7, 9, 10, 11, 12, 13, 14, 15, 20, 24, 27, 29, 30} {
    \fill[\scolor, fill opacity=\sopacity] (\x, \y) rectangle (\x+1, \y+1);
}
\def\yy{10}
\pgfmathsetmacro{\y}{int(30- \yy)}
\foreach \x in {1, 3, 5, 6, 7, 8, 9, 10, 12, 13, 14, 15, 19, 22, 23, 26, 28, 30} {
    \fill[\scolor, fill opacity=\sopacity] (\x, \y) rectangle (\x+1, \y+1);
}
\def\yy{11}
\pgfmathsetmacro{\y}{int(30- \yy)}
\foreach \x in {1, 2, 3, 4, 5, 8, 9, 11, 12, 13, 14, 15, 16, 21, 22, 23, 24} {
    \fill[\scolor, fill opacity=\sopacity] (\x, \y) rectangle (\x+1, \y+1);
}
\def\yy{12}
\pgfmathsetmacro{\y}{int(30- \yy)}
\foreach \x in {1, 2, 4, 6, 7, 8, 9, 10, 11, 12, 13, 15, 17, 21, 25, 26, 27} {
    \fill[\scolor, fill opacity=\sopacity] (\x, \y) rectangle (\x+1, \y+1);
}
\def\yy{13}
\pgfmathsetmacro{\y}{int(30- \yy)}
\foreach \x in {1, 2, 3, 4, 5, 7, 8, 10, 11, 13, 14, 15, 16, 18, 22, 25, 28, 29} {
    \fill[\scolor, fill opacity=\sopacity] (\x, \y) rectangle (\x+1, \y+1);
}
\def\yy{14}
\pgfmathsetmacro{\y}{int(30- \yy)}
\foreach \x in {1, 2, 3, 4, 5, 7, 8, 10, 11, 13, 14, 15, 18, 22, 25, 28, 29} {
    \fill[\scolor, fill opacity=\sopacity] (\x, \y) rectangle (\x+1, \y+1);
}
\def\yy{15}
\pgfmathsetmacro{\y}{int(30- \yy)}
\foreach \x in {2, 3, 5, 6, 7, 9, 10, 11, 12, 13, 14, 15, 20, 24, 27, 29, 30} {
    \fill[\scolor, fill opacity=\sopacity] (\x, \y) rectangle (\x+1, \y+1);
}
\def\yy{16}
\pgfmathsetmacro{\y}{int(30- \yy)}
\foreach \x in {1, 2, 3, 4, 5, 6, 7, 8, 9, 10, 11, 12, 16, 17, 18, 19, 20} {
    \fill[\scolor, fill opacity=\sopacity] (\x, \y) rectangle (\x+1, \y+1);
}
\def\yy{17}
\pgfmathsetmacro{\y}{int(30- \yy)}
\foreach \x in {1, 2, 3, 4, 5, 6, 7, 8, 9, 10, 11, 12, 16, 17, 18, 19, 20} {
    \fill[\scolor, fill opacity=\sopacity] (\x, \y) rectangle (\x+1, \y+1);
}
\def\yy{18}
\pgfmathsetmacro{\y}{int(30- \yy)}
\foreach \x in {1, 2, 3, 4, 5, 6, 7, 8, 9, 10, 11, 12, 16, 17, 18, 19, 20} {
    \fill[\scolor, fill opacity=\sopacity] (\x, \y) rectangle (\x+1, \y+1);
}
\def\yy{19}
\pgfmathsetmacro{\y}{int(30- \yy)}
\foreach \x in {1, 2, 3, 4, 5, 6, 7, 8, 9, 10, 11, 12, 16, 17, 18, 19, 20} {
    \fill[\scolor, fill opacity=\sopacity] (\x, \y) rectangle (\x+1, \y+1);
}
\def\yy{20}
\pgfmathsetmacro{\y}{int(30- \yy)}
\foreach \x in {1, 2, 3, 4, 5, 6, 7, 8, 9, 11, 12, 13, 14, 15, 16, 21, 22, 23, 24, 27} {
    \fill[\scolor, fill opacity=\sopacity] (\x, \y) rectangle (\x+1, \y+1);
}
\def\yy{21}
\pgfmathsetmacro{\y}{int(30- \yy)}
\foreach \x in {1, 2, 3, 4, 5, 6, 7, 8, 9, 11, 12, 13, 14, 15, 16, 21, 22, 23, 24, 27} {
    \fill[\scolor, fill opacity=\sopacity] (\x, \y) rectangle (\x+1, \y+1);
}
\def\yy{22}
\pgfmathsetmacro{\y}{int(30- \yy)}
\foreach \x in {1, 2, 3, 4, 5, 6, 7, 8, 9, 11, 12, 13, 14, 15, 16, 21, 22, 23, 24, 27} {
    \fill[\scolor, fill opacity=\sopacity] (\x, \y) rectangle (\x+1, \y+1);
}
\def\yy{23}
\pgfmathsetmacro{\y}{int(30- \yy)}
\foreach \x in {1, 2, 4, 6, 7, 8, 9, 10, 11, 12, 13, 15, 17, 21, 25, 26, 27, 30} {
    \fill[\scolor, fill opacity=\sopacity] (\x, \y) rectangle (\x+1, \y+1);
}
\def\yy{24}
\pgfmathsetmacro{\y}{int(30- \yy)}
\foreach \x in {1, 2, 4, 6, 7, 8, 9, 10, 11, 12, 13, 15, 17, 21, 25, 26, 27, 30} {
    \fill[\scolor, fill opacity=\sopacity] (\x, \y) rectangle (\x+1, \y+1);
}
\def\yy{25}
\pgfmathsetmacro{\y}{int(30- \yy)}
\foreach \x in {1, 2, 3, 4, 5, 6, 7, 8, 9, 11, 12, 13, 14, 15, 16, 20, 21, 22, 23, 24} {
    \fill[\scolor, fill opacity=\sopacity] (\x, \y) rectangle (\x+1, \y+1);
}
\def\yy{26}
\pgfmathsetmacro{\y}{int(30- \yy)}
\foreach \x in {1, 2, 3, 4, 5, 6, 7, 8, 9, 10, 11, 12, 16, 17, 18, 19, 20} {
    \fill[\scolor, fill opacity=\sopacity] (\x, \y) rectangle (\x+1, \y+1);
}
\def\yy{27}
\pgfmathsetmacro{\y}{int(30- \yy)}
\foreach \x in {1, 2, 3, 4, 5, 6, 7, 8, 9, 11, 12, 13, 14, 15, 16, 21, 22, 23, 24, 27} {
    \fill[\scolor, fill opacity=\sopacity] (\x, \y) rectangle (\x+1, \y+1);
}
\def\yy{28}
\pgfmathsetmacro{\y}{int(30- \yy)}
\foreach \x in {1, 2, 3, 4, 5, 6, 7, 8, 9, 10, 11, 12, 16, 17, 18, 19, 20} {
    \fill[\scolor, fill opacity=\sopacity] (\x, \y) rectangle (\x+1, \y+1);
}
\def\yy{29}
\pgfmathsetmacro{\y}{int(30- \yy)}
\foreach \x in {1, 2, 3, 4, 5, 6, 7, 8, 9, 10, 11, 12, 16, 17, 18, 19, 20, 29} {
    \fill[\scolor, fill opacity=\sopacity] (\x, \y) rectangle (\x+1, \y+1);
}
\def\yy{30}
\pgfmathsetmacro{\y}{int(30- \yy)}
\foreach \x in {1, 2, 4, 6, 7, 8, 9, 10, 11, 12, 13, 15, 17, 21, 25, 26, 27, 30} {
    \fill[\scolor, fill opacity=\sopacity] (\x, \y) rectangle (\x+1, \y+1);
}

\node at (13, 32) {\textbf{Messages received after two rounds of receiving messages}}; %
\node[rotate=90] at (-1, 26) {\textbf{Processes}}; %

\end{tikzpicture}
    \caption{Values possessed by each process after the second round of receiving messages from one of their $P \setminus F_i$. Also denoted as $T$ sets}
    \label{fig:counterexample-2}
\end{figure}
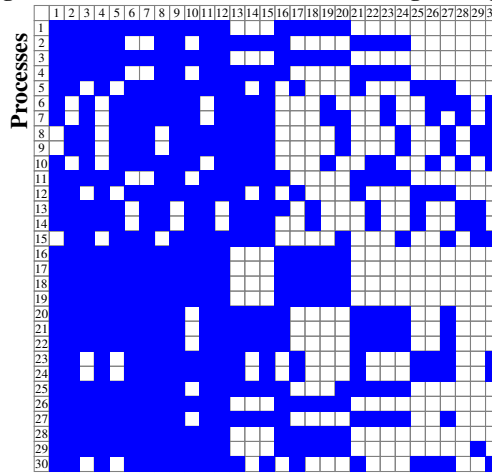

\begin{figure}[ht]
\centering
\begin{tikzpicture}[scale=0.2]
    \def\rows{30}
    \def\cols{30}

    \foreach \x in {0,...,\cols} {
        \foreach \y in {0,...,\rows} {
            \draw[gray] (\x, \y) rectangle (\x+1, \y+1);
        }
    }

    \foreach \x in {1,...,30} {
        \pgfmathsetmacro{\result}{int(31- \x)} 
        \node[align=center, outer sep=0pt] at (\x + 0.5, 30.5) {\fontsize{6}{0}\selectfont \textcolor{black}{\scalebox{.8}{\x}}};
    }
    \foreach \y in {1,...,30} {
        \pgfmathsetmacro{\result}{int(31- \y)} 
        \node[align=center, text width=1cm] at (0.5, \y - 0.5) {\fontsize{6}{0}\selectfont \textcolor{black}{\scalebox{.8}{\result}}};
    }

\def\yy{1}
\pgfmathsetmacro{\y}{int(30- \yy)}
\foreach \x in {1, 2, 3, 4, 5, 6, 7, 8, 9, 10, 11, 12, 13, 14, 15, 16, 17, 18, 19, 20, 21, 22, 23, 24, 25, 26, 27} {
    \fill[\scolor, fill opacity=\sopacity] (\x, \y) rectangle (\x+1, \y+1);
}
\def\yy{2}
\pgfmathsetmacro{\y}{int(30- \yy)}
\foreach \x in {1, 2, 3, 4, 5, 6, 7, 8, 9, 10, 11, 12, 13, 14, 15, 16, 17, 18, 19, 20, 23, 24, 26, 27, 28, 29, 30} {
    \fill[\scolor, fill opacity=\sopacity] (\x, \y) rectangle (\x+1, \y+1);
}
\def\yy{3}
\pgfmathsetmacro{\y}{int(30- \yy)}
\foreach \x in {1, 2, 3, 4, 5, 6, 7, 8, 9, 10, 11, 12, 13, 14, 15, 16, 17, 18, 19, 20, 21, 22, 23, 24, 25, 26, 27} {
    \fill[\scolor, fill opacity=\sopacity] (\x, \y) rectangle (\x+1, \y+1);
}
\def\yy{4}
\pgfmathsetmacro{\y}{int(30- \yy)}
\foreach \x in {1, 2, 3, 4, 5, 6, 7, 8, 9, 10, 11, 12, 13, 14, 15, 16, 17, 18, 19, 20, 23, 24, 26, 27, 28, 29, 30} {
    \fill[\scolor, fill opacity=\sopacity] (\x, \y) rectangle (\x+1, \y+1);
}
\def\yy{5}
\pgfmathsetmacro{\y}{int(30- \yy)}
\foreach \x in {1, 2, 3, 4, 5, 6, 7, 8, 9, 10, 11, 12, 13, 14, 15, 16, 17, 19, 21, 22, 23, 24, 25, 26, 27, 28, 30} {
    \fill[\scolor, fill opacity=\sopacity] (\x, \y) rectangle (\x+1, \y+1);
}
\def\yy{6}
\pgfmathsetmacro{\y}{int(30- \yy)}
\foreach \x in {1, 2, 3, 4, 5, 6, 7, 8, 9, 10, 11, 12, 13, 14, 15, 16, 18, 20, 21, 22, 23, 24, 25, 27, 28, 29, 30} {
    \fill[\scolor, fill opacity=\sopacity] (\x, \y) rectangle (\x+1, \y+1);
}
\def\yy{7}
\pgfmathsetmacro{\y}{int(30- \yy)}
\foreach \x in {1, 2, 3, 4, 5, 6, 7, 8, 9, 10, 11, 12, 13, 14, 15, 16, 18, 20, 21, 22, 23, 24, 25, 27, 28, 29, 30} {
    \fill[\scolor, fill opacity=\sopacity] (\x, \y) rectangle (\x+1, \y+1);
}
\def\yy{8}
\pgfmathsetmacro{\y}{int(30- \yy)}
\foreach \x in {1, 2, 3, 4, 5, 6, 7, 8, 9, 10, 11, 12, 13, 14, 15, 17, 18, 20, 21, 22, 24, 25, 26, 27, 28, 29, 30} {
    \fill[\scolor, fill opacity=\sopacity] (\x, \y) rectangle (\x+1, \y+1);
}
\def\yy{9}
\pgfmathsetmacro{\y}{int(30- \yy)}
\foreach \x in {1, 2, 3, 4, 5, 6, 7, 8, 9, 10, 11, 12, 13, 14, 15, 17, 18, 20, 21, 22, 24, 25, 26, 27, 28, 29, 30} {
    \fill[\scolor, fill opacity=\sopacity] (\x, \y) rectangle (\x+1, \y+1);
}
\def\yy{10}
\pgfmathsetmacro{\y}{int(30- \yy)}
\foreach \x in {1, 2, 3, 4, 5, 6, 7, 8, 9, 10, 11, 12, 13, 14, 15, 16, 18, 20, 21, 22, 23, 24, 25, 27, 28, 29, 30} {
    \fill[\scolor, fill opacity=\sopacity] (\x, \y) rectangle (\x+1, \y+1);
}
\def\yy{11}
\pgfmathsetmacro{\y}{int(30- \yy)}
\foreach \x in {1, 2, 3, 4, 5, 6, 7, 8, 9, 10, 11, 12, 13, 14, 15, 16, 17, 18, 19, 20, 23, 24, 26, 27, 28, 29, 30} {
    \fill[\scolor, fill opacity=\sopacity] (\x, \y) rectangle (\x+1, \y+1);
}
\def\yy{12}
\pgfmathsetmacro{\y}{int(30- \yy)}
\foreach \x in {1, 2, 3, 4, 5, 6, 7, 8, 9, 10, 11, 12, 13, 14, 15, 16, 17, 19, 21, 22, 23, 24, 25, 26, 27, 28, 30} {
    \fill[\scolor, fill opacity=\sopacity] (\x, \y) rectangle (\x+1, \y+1);
}
\def\yy{13}
\pgfmathsetmacro{\y}{int(30- \yy)}
\foreach \x in {1, 2, 3, 4, 5, 6, 7, 8, 9, 10, 11, 12, 13, 14, 15, 16, 17, 18, 19, 20, 22, 23, 25, 26, 28, 29, 30} {
    \fill[\scolor, fill opacity=\sopacity] (\x, \y) rectangle (\x+1, \y+1);
}
\def\yy{14}
\pgfmathsetmacro{\y}{int(30- \yy)}
\foreach \x in {1, 2, 3, 4, 5, 6, 7, 8, 9, 10, 11, 12, 13, 14, 15, 16, 17, 18, 19, 20, 22, 23, 25, 26, 28, 29, 30} {
    \fill[\scolor, fill opacity=\sopacity] (\x, \y) rectangle (\x+1, \y+1);
}
\def\yy{15}
\pgfmathsetmacro{\y}{int(30- \yy)}
\foreach \x in {1, 2, 3, 4, 5, 6, 7, 8, 9, 10, 11, 12, 13, 14, 15, 17, 18, 20, 21, 22, 24, 25, 26, 27, 28, 29, 30} {
    \fill[\scolor, fill opacity=\sopacity] (\x, \y) rectangle (\x+1, \y+1);
}
\def\yy{16}
\pgfmathsetmacro{\y}{int(30- \yy)}
\foreach \x in {1, 2, 3, 4, 5, 6, 7, 8, 9, 10, 11, 12, 13, 14, 15, 16, 17, 18, 19, 20, 21, 22, 23, 24, 25, 26, 27} {
    \fill[\scolor, fill opacity=\sopacity] (\x, \y) rectangle (\x+1, \y+1);
}
\def\yy{17}
\pgfmathsetmacro{\y}{int(30- \yy)}
\foreach \x in {1, 2, 3, 4, 5, 6, 7, 8, 9, 10, 11, 12, 13, 14, 15, 16, 17, 18, 19, 20, 21, 22, 23, 24, 25, 26, 27} {
    \fill[\scolor, fill opacity=\sopacity] (\x, \y) rectangle (\x+1, \y+1);
}
\def\yy{18}
\pgfmathsetmacro{\y}{int(30- \yy)}
\foreach \x in {1, 2, 3, 4, 5, 6, 7, 8, 9, 10, 11, 12, 13, 14, 15, 16, 17, 18, 19, 20, 21, 22, 23, 24, 25, 26, 27} {
    \fill[\scolor, fill opacity=\sopacity] (\x, \y) rectangle (\x+1, \y+1);
}
\def\yy{19}
\pgfmathsetmacro{\y}{int(30- \yy)}
\foreach \x in {1, 2, 3, 4, 5, 6, 7, 8, 9, 10, 11, 12, 13, 14, 15, 16, 17, 18, 19, 20, 21, 22, 23, 24, 25, 26, 27} {
    \fill[\scolor, fill opacity=\sopacity] (\x, \y) rectangle (\x+1, \y+1);
}
\def\yy{20}
\pgfmathsetmacro{\y}{int(30- \yy)}
\foreach \x in {1, 2, 3, 4, 5, 6, 7, 8, 9, 10, 11, 12, 13, 14, 15, 16, 17, 18, 19, 20, 21, 22, 23, 24, 26, 27, 28, 29, 30} {
    \fill[\scolor, fill opacity=\sopacity] (\x, \y) rectangle (\x+1, \y+1);
}
\def\yy{21}
\pgfmathsetmacro{\y}{int(30- \yy)}
\foreach \x in {1, 2, 3, 4, 5, 6, 7, 8, 9, 10, 11, 12, 13, 14, 15, 16, 17, 18, 19, 20, 21, 22, 23, 24, 26, 27, 28, 29, 30} {
    \fill[\scolor, fill opacity=\sopacity] (\x, \y) rectangle (\x+1, \y+1);
}
\def\yy{22}
\pgfmathsetmacro{\y}{int(30- \yy)}
\foreach \x in {1, 2, 3, 4, 5, 6, 7, 8, 9, 10, 11, 12, 13, 14, 15, 16, 17, 18, 19, 20, 21, 22, 23, 24, 26, 27, 28, 29, 30} {
    \fill[\scolor, fill opacity=\sopacity] (\x, \y) rectangle (\x+1, \y+1);
}
\def\yy{23}
\pgfmathsetmacro{\y}{int(30- \yy)}
\foreach \x in {1, 2, 3, 4, 5, 6, 7, 8, 9, 10, 11, 12, 13, 14, 15, 16, 17, 19, 21, 22, 23, 24, 25, 26, 27, 28, 30} {
    \fill[\scolor, fill opacity=\sopacity] (\x, \y) rectangle (\x+1, \y+1);
}
\def\yy{24}
\pgfmathsetmacro{\y}{int(30- \yy)}
\foreach \x in {1, 2, 3, 4, 5, 6, 7, 8, 9, 10, 11, 12, 13, 14, 15, 16, 17, 19, 21, 22, 23, 24, 25, 26, 27, 28, 30} {
    \fill[\scolor, fill opacity=\sopacity] (\x, \y) rectangle (\x+1, \y+1);
}
\def\yy{25}
\pgfmathsetmacro{\y}{int(30- \yy)}
\foreach \x in {1, 2, 3, 4, 5, 6, 7, 8, 9, 10, 11, 12, 13, 14, 15, 16, 17, 18, 19, 20, 21, 22, 23, 24, 26, 27, 28, 29, 30} {
    \fill[\scolor, fill opacity=\sopacity] (\x, \y) rectangle (\x+1, \y+1);
}
\def\yy{26}
\pgfmathsetmacro{\y}{int(30- \yy)}
\foreach \x in {1, 2, 3, 4, 5, 6, 7, 8, 9, 10, 11, 12, 13, 14, 15, 16, 17, 18, 19, 20, 21, 22, 23, 24, 25, 26, 27} {
    \fill[\scolor, fill opacity=\sopacity] (\x, \y) rectangle (\x+1, \y+1);
}
\def\yy{27}
\pgfmathsetmacro{\y}{int(30- \yy)}
\foreach \x in {1, 2, 3, 4, 5, 6, 7, 8, 9, 10, 11, 12, 13, 14, 15, 16, 17, 18, 19, 20, 21, 22, 23, 24, 26, 27, 28, 29, 30} {
    \fill[\scolor, fill opacity=\sopacity] (\x, \y) rectangle (\x+1, \y+1);
}
\def\yy{28}
\pgfmathsetmacro{\y}{int(30- \yy)}
\foreach \x in {1, 2, 3, 4, 5, 6, 7, 8, 9, 10, 11, 12, 13, 14, 15, 16, 17, 18, 19, 20, 21, 22, 23, 24, 25, 26, 27} {
    \fill[\scolor, fill opacity=\sopacity] (\x, \y) rectangle (\x+1, \y+1);
}
\def\yy{29}
\pgfmathsetmacro{\y}{int(30- \yy)}
\foreach \x in {1, 2, 3, 4, 5, 6, 7, 8, 9, 10, 11, 12, 13, 14, 15, 16, 17, 18, 19, 20, 21, 22, 23, 24, 25, 26, 27, 29} {
    \fill[\scolor, fill opacity=\sopacity] (\x, \y) rectangle (\x+1, \y+1);
}
\def\yy{30}
\pgfmathsetmacro{\y}{int(30- \yy)}
\foreach \x in {1, 2, 3, 4, 5, 6, 7, 8, 9, 10, 11, 12, 13, 14, 15, 16, 17, 19, 21, 22, 23, 24, 25, 26, 27, 28, 30} {
    \fill[\scolor, fill opacity=\sopacity] (\x, \y) rectangle (\x+1, \y+1);
}

\node at (13, 32) {\textbf{Messages received after three rounds of receiving messages}}; %
\node[rotate=90] at (-1, 26) {\textbf{Processes}}; %

\end{tikzpicture}
    \caption{Values possessed by each process after the third round of receiving messages from the processes in $P \setminus F_i$. Also denotes as $U$ sets. Observe that there is no set $S^+$ (where $S^+$ is a quorum of a wise process) that is received by all processes. This happens because all quorums of all processes contain at least one element in the range $[16, 30]$ and as can be observed in the picture, all processes are missing at least one element from this range in their received values. }
    \label{fig:counterexample-3}
\end{figure}
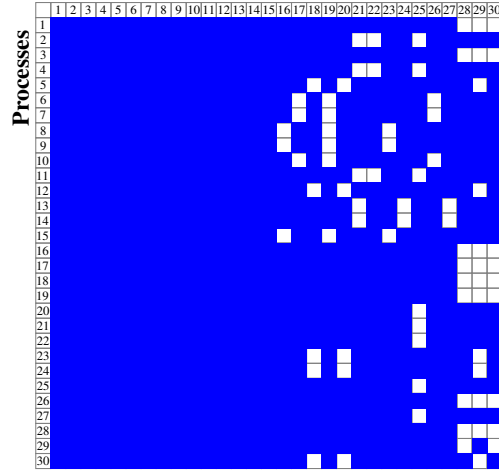

\clearpage

\bibliography{references, dblpbibtex}
\bibliographystyle{ACM-Reference-Format}

\end{document}